\documentclass[journal]{IEEEtran}
\IEEEoverridecommandlockouts
\ifCLASSINFOpdf
\else
\fi

\usepackage{tabularx}
\usepackage{threeparttable}
\usepackage{mathrsfs}
\usepackage{color}
\usepackage{amsfonts,amssymb}
\usepackage{amsfonts}
\usepackage{subfigure}
\usepackage{booktabs}

\usepackage{amsmath}
\usepackage{enumerate}
\usepackage{algorithm}
\usepackage{algorithmic}
\usepackage{bm}
\usepackage{makecell}
\usepackage{multirow}
\usepackage{mathtools}
\usepackage{bbm}
\usepackage{dsfont}
\usepackage{pifont}
\usepackage{amsthm}
\usepackage{cite}
\usepackage{balance}
\usepackage{graphicx}

\newtheorem{remark}{Remark}
\newtheorem{theorem}{Theorem}
\newtheorem{lemma}{Lemma}
\newtheorem{corollary}{Corollary}
\newtheorem{assumption}{Assumption}

\begin{document}

\title{
Preserving Topology of Network Systems:\\ Metric, Analysis, and Optimal Design
}
\author{Yushan Li, Zitong Wang, Jianping He, Cailian Chen and Xinping Guan
\thanks{The authors are with Dept. of Automation, Shanghai Jiao Tong University, Key Laboratory of System Control and Information Processing, Ministry of Education of China, and Shanghai Engineering Research Center of Intelligent Control and Management, Shanghai, China. E-mail address: \{yushan\_li, wangzitong, jphe, cailianchen, xpguan\}@sjtu.edu.cn. 
}%
}

\maketitle

\begin{abstract}
Preserving the topology from being inferred by external adversaries has become a paramount security issue for network systems (NSs), and adding random noises to the nodal states provides a promising way. 
Nevertheless, recent works have revealed that the topology cannot be preserved under i.i.d. noises in the asymptotic sense. 
How to effectively characterize the non-asymptotic preservation performance still remains an open issue. 
Inspired by the deviation quantification of concentration inequalities, this paper proposes a novel metric named trace-based variance-expectation ratio. 
This metric effectively captures the decaying rate of the topology inference error, where a slower rate indicates better non-asymptotic preservation performance. 
We prove that the inference error will always decay to zero asymptotically, as long as the added noises are non-increasing and independent (milder than the i.i.d. condition). 
Then, the optimal noise design that produces the slowest decaying rate for the error is obtained. 
More importantly, we amend the noise design by introducing one-lag time dependence, achieving the zero state deviation and the non-zero topology inference error in the asymptotic sense simultaneously. 
Extensions to a general class of noises with multi-lag time dependence are provided. 
Comprehensive simulations verify the theoretical findings. 

\end{abstract}


\IEEEpeerreviewmaketitle

\section{Introduction}

Network Systems (NSs) have received considerable attention in last decades due to their wide applications, e.g., robotic, social, and power networks, to name a few \cite{olfati2007consensus}. 
Among these applications, the topology structure between the nodes (or agents) specifies the information flow of the internal interaction and fundamentally affects the cooperation performance of NSs \cite{kar2014asymptotically,brugere2018network}. 
Unfortunately, in many scenarios, the sensitive topology information is at a great risk of being inferred by external adversaries \cite{dong2019learning}. 
For instance, the topology of a team of mobile robots can be inferred from the observed robot trajectories in the physical space \cite{vasquez2018network,li2023local}, 
and the topology of related social users can also be inferred from the publicly released data \cite{matta2020learning,lu2020privacy}. 
With the topology mastered, the adversary can further launch more precise and intelligent attacks against critical nodes in NSs, causing severe breaches and damages \cite{han2018privacy,katewa2020differential,sandberg2022secure}. 
Therefore, it is of vital importance to preserve the topology from being inferred for security concerns.

Mathematically, the topology preservation can be seen as a maxmin version of the topology inference \cite{wang2022distributed}, where the latter one has been widely investigated in the literature. 
For example, \cite{egilmez2017graph,zhu2020network} utilizes the techniques of graph signal processing to capture the underlying undirected topology among different variables. 
\cite{santos2019local,cirillo2023estimating} consider learning the topology of NSs governed by a linear diffusion process in a local observability regime. 
The vector autoregressive analysis is also a popular tool to deal with the directed topology inference problem \cite{ioannidis2019semi,zaman2021online}. 
The readers can refer to \cite{giannakis2018topology,mateos2019connecting,matta2020graph} for a detailed survey. 

It is straightforward that if the interaction process of a NS is noise-free or the adversary can access all data associated with the process, then with some basic identifiability conditions, the topology can be readily inferred by a group of equations from the observation groups \cite{vanwaarde2021topology}. 
Hence, a promising and effective way for topology preservation is to add random noises to the states of each node during their interactions. 
In the last decade, the noise-adding mechanism has been used to preserve the data privacy while ensuring the global computation goal of NSs in the literature, e.g., see \cite{mo2017privacy,nozari2017differentially,he2018privacy,wang2021dynamic}. 
Since the topology can only be inferred based on the noisy states, it is sure that the noise-adding mechanism will achieve certain topology preservation effects compared with the noise-free case. 
The works \cite{he2018preserving,wang2019privacya,he2020differential} have further developed mature techniques for the privacy analysis and noise design to preserve the state privacy of NSs. 
However, they are hard to be extended on the topology preservation due to the following reasons. 
First, although the inferred topology can be represented by the noisy states, it is essentially a highly nonlinear mapping of the states, making it intractable to analyze the preservation performance as doing that for the state itself. 
Second, concerning the noise design for optimal topology preservation, the analytical procedures are further complicated considering the various distribution forms and the accumulation impact of the noises. 

Recently, the \textit{sample complexity}\footnote{In the literature, the sample complexity usually means the observation/data amount needed to identify the unknown system parameters within prescribed confidence levels.} analysis of ordinary least squares (OLS) estimators in the system identification community has shed light on characterizing the topology inference performance \cite{simchowitz2018learning,sarkar2019near,sun2020finite,zheng2021nonasymptotic}. 
In the context of topology inference, the above sample complexity can be interpreted as a disclosure probability that the adversary can infer the topology within a given small interval to quantify the inference accuracy \cite{jedra2023finitetime}. 
Nevertheless, this metric is hard to be reversely leveraged to find the optimal noise design for topology preservation, 
because i) the metric is mainly for evaluating the estimation performance and contains many implicit parameters that are inconvenient for the computation and design of noises; 
ii) the theoretical framework relies on the fundamental assumption that the added noises are i.i.d., which hinders us from applying the analysis procedures to design non-i.i.d. noises for topology preservation concerns 
(we refer to \cite{tsiamis2022statistical} for a detailed review about the sample complexity analysis). 
In a nutshell, it is still an open issue to find an appropriate metric to describe the topology preservation performance and use it to design the optimal noises accordingly.


Motivated by the above issues, this paper investigates the performance metric and noise design problem for preserving the topology of NSs. 
The challenges lie in two aspects. 
First, the metric is desired to have no dependency on the distribution forms of the added noises, and enjoy the same rate characterization concerning the observation number as existing results in sample complexity analysis. 
Second, taking the metric as the guidance for the noise design, the added noises should ensure the topology cannot be accurately inferred while not jeopardizing the normal state evolution. 
To address these issues, the key insight is to extract and characterize the convergence rate of the topology inference error regarding the observation number.  
The contributions of our work are summarized as follows. 
\begin{itemize}
\item (\textit{Metric}) We advocate a novel metric named trace-based variance-expectation ratio to characterize the topology inference error. 
This metric is derived by applying concentration inequalities on the error factors and exploiting their deviation quantification, and enjoys efficient computation. 
Specifically, we prove that the ratio exhibits the same decaying rate of the inference error under i.i.d. noises as the sample complexity does.

\item (\textit{Optimal independent noises}) Based on the proposed metric, we provide the thorough analysis of the topology preservation under time-independent non-increasing noises (milder than i.i.d. condition), by resorting to the Euler-Maclaurin approximation. 
We obtain the explicit decaying rate of the inference error and reveal that the topology can never be preserved asymptotically. 
Then, we give the optimal noise design to achieve the slowest decaying rate for the non-asymptotic preservation. 

\item (\textit{Asymptotic preservation}) We extend the proposed metric to a more general situation, where the designed noises are dependent on each other in time. 
Under the conditions that the noise variance decays to zero in a polynomial rate, we derive that the state deviation with the noise-free one will decay to zero in the mean square sense, while the topology inference error will remain a constant even in the asymptotic sense. 
Extensive simulations corroborate the theoretical results.

\end{itemize}

In summary, the proposed metric well accommodates the rate characterization of the sample complexity under i.i.d. noises, and can  be further used to design more general forms of noises for topology preservation and analyze the performance explicitly. 
The remainder of this paper is organized as follows. 
Section \ref{sec:preliminary} gives basic preliminaries and formulates the problem of interest. 
The proposed metric is presented in Section \ref{sec:metric}. 
The analysis for topology preservation under independent noises is given in Section \ref{sec:independent}. 
In Section \ref{sec:dependent}, we extend the metric to analyze topology preservation performance under dependent noises. 
Simulation results are shown in Section \ref{sec:simulation}. 
Concluding remarks and further research issues are given in Section \ref{sec:conclusion}.

\section{Preliminaries and Problem Formulation}\label{sec:preliminary}

\subsection{Graph Basics and Notations}
Let $\mathcal{G}=(\mathcal{V},\mathcal{E})$ be a directed graph that models the network system, where $\mathcal{V}=\{1, \cdots, n\}$ is the finite set of nodes and $\mathcal{E}\subseteq \mathcal{V}\times \mathcal{V}$ is the set of interaction edges.
An edge $(i,j)\in \mathcal{E}$ indicates that $i$ will use information sent by $j$.
The adjacency matrix $\mathcal{A}=[a_{ij}]_{n \times n}$ of $\mathcal{G}$ is defined such that ${a}_{ij}\!>\!0$ if $(i,j)$ exists, and ${a}_{ij}\!=\!0$ otherwise.
Denote ${\mathcal{N}_i}=\{j\in \mathcal{V}:a_{ij}>0\}$ as the in-neighbor set of $i$, and $d_i=\left| {\mathcal{N}_i} \right |$ as its in-degree. 
A directed graph has a (directed) spanning tree if there exists at least a node having a directed path to all other nodes. 

Throughout the paper, the notations $\rho(M_0)$, $\| M_0 \|$, $\| M_0  \|_{F}$ and $\operatorname{tr}(M_0)$ represent the spectral radius, the spectral norm, the Frobenius norm, and the trace of a square matrix $M_0$, respectively. 
The superscript $\cdot^\mathsf{T}$ indicates the transpose of a matrix. 
Specifically, it holds that 
\begin{align}
\rho(M_0)\le \| M_0 \| \le \| M_0  \|_{F}=\sqrt{\operatorname{tr}( M_0^\mathsf{T} M_0 )}  . 
\end{align}
We use $\mathbb{E}[ M_0]$ and $\mathbb{D}[ M_0]$ to represent the element-wise expectation and variance of $M_0$, respectively. 
Let $\mathbb{N}$ and $\mathbb{N}^+$ be the set of natural numbers and the set of positive integers. 
For two real-valued functions $f_1$ and $f_2$, $f_1(x)=\bm{O}(f_2(x))$ as $x\to x_0$ means $\mathop {\lim }\nolimits_{x \to x_0 } |f_1(x)/f_2(x)|<\infty$.

\subsection{System Model}
Consider that each node in the NS updates its state by the following model
\begin{equation}\label{eq:single_nodal_model}
\begin{aligned}
x_{t+1}^i =  x_t^i + \sum_{j\in \mathcal{N}_{i}} w_{i j} (x_t^j-x_t^i) +\theta_t^i,
\end{aligned}
\end{equation}
where $x_t^i\in\mathbb{R}$ is node $i$'s state at time $t$ ($t=0,1,\cdots,T$), $w_{i j}\in [0,1]$ is the interaction weight satisfying $w_{ij}>0$ if the edge $(i,j)\in \mathcal{E}$ exists or $w_{i j}=0$ otherwise ($i\neq j$), and $\theta_t^i$ is the added noise by node $i$ for security concerns. 
Note that many mature methods can be found to design the interaction weight $w_{i j}$ in the literature. 
For instance, by the Laplacian rule \cite{sayed2014adaptation}, $w_{ij}$ is given by 
\begin{align}\label{eq:topo-rule}
w_{ij}=\left \{
\begin{aligned}
&\gamma a_{ij}/ d_{\max},&&\text{if}~i\neq j \\
&1-\sum\nolimits_{j\in \mathcal{N}_{i}} w_{ij}, &&\text{if}~i\neq j
\end{aligned}
\right.,
\end{align}
where $d_{\max}=\max\{|\mathcal{N}_i|,i\in\mathcal{V}\}$ is the maximum in-degree of the NS, and $\gamma\in(0,1]$ is an auxiliary parameter. 

In a global form, all the nodal models are integrated as 
\begin{align}\label{eq:global_model}
x_{t+1}=W x_t + \theta_t,
\end{align}
where $x_t=[x_t^1,\cdots,x_t^n]^\mathsf{T}\in\mathbb{R}^n$, $\theta_t=[\theta_t^1,\cdots,\theta_t^n]^\mathsf{T}\in\mathbb{R}^n$, and $W=[w_{ij}]_{i,j=1}^{n}\in \mathbb{R}^{n\times n}$. 
As for the topology setting of the NS, the following assumption is made. 
\begin{assumption}\label{assu:topo}
The graph $\mathcal{G}$ has at least a spanning tree, and the associated topology $W$ is non-negative and row-stochastic. 
\end{assumption}

It is worth mentioning that Assumption \ref{assu:topo} is milder than the commonly used assumption that the graph is connected and the topology matrix $W$ is doubly-stochastic in the literature. 
When Assumption \ref{assu:topo} holds, $W$ is ensured to be primitive (i.e., there exists $t\in \mathbb{N}$ such that all elements of $W^t$ are positive). 
Let $\lambda_i$ be the $i$-th eigenvalue of $W$, ordered by $ |\lambda _1| \geq |\lambda _2| \geq \cdots \geq |\lambda _N|$, 
and it follows from the famous Perron-Frobenius theorem (see \cite[Theorem 2.12]{FB-LNS}) that
\begin{align}\label{eq:lambda}
\lambda_1=1>|\lambda _i|,~\forall i\in \{2,\cdots,n\}. 
\end{align}
The property \eqref{eq:lambda} is crucial for the performance analysis in the subsequent sections. 
The following variables are introduced for simple expressions throughout this paper
\begin{align}
&\Gamma_t=\sum_{m = 0}^{t-1} W^m (W^m)^\mathsf{T},~\Gamma_t^* =\sum_{m = 0}^{t-1} (W^m)^\mathsf{T} W^m, \nonumber \\
& \bm{\theta}_{0:T-1}\!=\![\theta_0^\mathsf{T},\theta_1^\mathsf{T},\!\cdots\!,\theta_{T-1}^\mathsf{T}]^\mathsf{T},~\Theta(T) =[\theta_0,\theta_1,\cdots,~\theta_{T-1}], \nonumber\\
&X(T)=[x_{0},x_{2},\cdots,x_{T-1}], ~X^+(T) =[x_{1},x_{2},\cdots,x_{T}], \nonumber
\end{align}
where the parentheses notation $(T)$ is omitted in the following analysis when no confusion is caused.

\subsection{Topology Inference and Preservation}
Consider that an adversary can observe the system state and aims to infer the topology matrix $W$, which can be leveraged to support further intelligent attacks against the NS. 
Based on the collected observations, 
the topology $W$ can be inferred by the following OLS estimator \cite{santos2019local}
\begin{align}\label{eq:topo_estimator}
\hat{W}(T)= X^+ X^\mathsf{T} (X X^\mathsf{T})^{-1}.  
\end{align}
Accordingly, the inference error of \eqref{eq:topo_estimator} is given by 
\begin{align}\label{eq:topo_error}
\|\hat{W}(T)-W\|= \| \Theta X^\mathsf{T} (X X^\mathsf{T})^{-1} \|. 
\end{align}

As a countermeasure, the goal of topology preservation methods is to design the added noises $\{\theta_t\}_{t=0}^{T-1}$ appropriately such that the inference accuracy of $\hat{W}(T)$ is deteriorated, which can be formulated as solving the following problem
\begin{subequations}\label{eq:optimization_problem1}
\begin{align}
\textbf{P}_\textbf{0}:~~~~\mathop{\max }\limits_{ \Theta(T)  }~~& \| \hat{W}(T) - W \|  \label{eq:optim-a}\\
{\rm{s.t.}}~~&  \lim \limits_{t\rightarrow \infty } \| x_t - x^*_t \| = 0, \label{eq:optim-c}
\end{align}
\end{subequations}
where $x^*_t=W^t x_0$ represents the ideal noise-free system state at time $T$. 
Note that the constraint \eqref{eq:optim-c} is required to ensure that the deviation between the actual and ideal states of the NS decays to asymptotically. 
It should be pointed out that finding the explicit optimal solution to the problem $\textbf{P}_\textbf{0}$ is intractable due to the high nonlinearity and long optimization horizon about $\Theta$. 
The latest work \cite{wang2022distributed} provides a sequential algorithm to obtain the optimal noise magnitude at each iteration, but it is still not global optimal.

\subsection{Problem of Interests}
In this paper, we aim to analyze the topology preservation performance by adding random noises and guide the optimal noise design. 
Note that when the noises are i.i.d., existing works have well characterized the inference error from the perspective of the following sample complexity \cite{jedra2023finitetime}
\begin{align}\label{eq:sample_complexity}
\Pr\left( \| \hat{W}(T) - W \|<\varepsilon(\delta,T) \right)\ge 1- \delta, 
\end{align}
where $\varepsilon(\delta,T)\ge0$ is the inference accuracy that decreases with $T$ growing, and $\delta\in[0,1]$ is the confidence level about $\hat{W}(T) $. 
Specifically, by appropriately setting $\delta$ into a decaying-to-zero form regarding $T$, one can further obtain that 
\begin{align}
\mathop {\lim }\limits_{T \to \infty} \Pr\left( \| \hat{W}(T) - W \|=0 \right)=1,
\end{align}
which implies that simply maximizing $\| \hat{W}(T) - W \|$ may not preserve the topology from being inferred asymptotically. 
Motivated by this point, we highlight that it would be much more meaningful to make $\| \hat{W}(T) - W \|$ approach to zero as slow as possible. 
This decaying rate effectively characterizes the non-asymptotic performance of $\| \hat{W}(T) - W \|$ from the preservation perspective. 
Therefore, instead of focusing on $\textbf{P}_\textbf{0}$, the goal of our work can be formulated as the following rate-maximization problem
\begin{subequations}\label{eq:optimization_problem2}
\begin{align}
\textbf{P}_\textbf{1}:~~~~\mathop{\max }\limits_{ \Theta(T)  }~~&  R\left(\| \hat{W}(T) - W \| \right) \label{eq:optim2-a}\\
{\rm{s.t.}}~~&  \lim \limits_{t\rightarrow \infty } \| x_t - x^*_t \| = 0, \label{eq:optim2-c}
\end{align}
\end{subequations}
where $R(\cdot)$ is the \textit{dominating rate factor}\footnote{For instance, given two functions $\frac{c_1}{T^3}$ and $c_2 \rho^T$ ($0\!<\rho\!<\!1$), we have $R(\frac{c_1}{T^3})=\bm{O}(\frac{1}{T^3})$ and $R(c_2 \rho^T)=\bm{O}(\rho^T)$. Specifically, the rate $R(c_2 \rho^T)$ is faster than $R(\frac{c_1}{T^3})$ due to $\lim \limits_{T\rightarrow \infty }{R(c_2 \rho^T)}/{R(\frac{c_1}{T^3})}=0$.} that is determined by $T$, and can be optimized globally. 
Based on the above formulation, we are interested in the following three aspects.
\begin{itemize}
\item How to establish an effective and computation-tractable metric to characterize $R(\| \hat{W}(T) \!-\! W \| )$, which needs to accommodate non-i.i.d. noises regardless of the distribution forms. 
\item If such a performance metric is found, use this metric to analyze the topology preservation performance of typical noise-adding methods and give the explicit expressions. 
\item Finding the (near-)optimal noise design that achieves reliable topology preservation performance even in the asymptotic case. 
\end{itemize}



\section{Metric Design Under Independent Noises}\label{sec:metric}

In this section, a trace-based variance-expectation ratio metric is proposed for the case of adding independent noises, and its relationship with the sample complexity in terms of the decaying rate is analyzed.

\subsection{Metric Design Inspired From Concentration Inequalities}

Notice that due to the added random noises, the matrix $X X^\mathsf{T}$ is invertible (non-singular) almost surely according to Sard’s theorem in measure theory. 
Hence, by the triangle inequality, the inference error is bounded by 
\begin{align}
\|\hat{W}(T) \!-\! W\|  \!\le\! \| \Theta X^\mathsf{T} \| \| (X X^\mathsf{T})^{-1} \| \!=\!  \kappa(X X^\mathsf{T}) \frac{ \| \Theta X^\mathsf{T} \|  }{ \|X X^\mathsf{T}\| },
\end{align}
where $ \kappa(X X^\mathsf{T}) \!<\! \infty $ is the conditional number of $X X^\mathsf{T}$. 
Clearly, the convergence rate of $\|\hat{W}(T)-W\|$ regarding $T$ is determined by the ratio $\frac{ \| \Theta X^\mathsf{T} \|  }{ \|X X^\mathsf{T}\| }$. 
Based on the ratio structure, we first present the following metric to characterize $R(\| \hat{W}(T) - W \|) $ under independent noises: 
\begin{itemize}
\item \textbf{The variance-expectation ratio under independent noises}:
\begin{equation}\label{eq:def_trace}
R_{\theta}(T)=\frac{ \sqrt{ \operatorname{tr}(\mathbb{D}\left[\Theta (X -\mathbb{E}[X])^\mathsf{T}  \right]) } }{ \operatorname{tr}(\mathbb{E}[(X -\mathbb{E}[X]) (X -\mathbb{E}[X])^\mathsf{T}] )},
\end{equation}
\end{itemize}
where the term $\mathbb{E}[X]$ is subtracted by $X$ to accommodate arbitrary bound initial state $x_0$.  

Next, we provide the detailed analysis of how $R_{\theta}(T)$ is motivated to be designed. 
For simple and intuitive interpretation, here we suppose $x_0=0$ as the sample complexity analysis does, and use the diagonal elements  $(\Theta X^\mathsf{T})^{ii}$ and $(X X^\mathsf{T})^{ii}$ as examples. 
When the noises $\{ \theta_t\}_{t=0}^{T-1}$ are independent in time, $\mathbb{E}[(\Theta X^\mathsf{T})^{ii}] =0 $ and $ \mathbb{E}[(X X^\mathsf{T})^{ii}] >0 $ hold. 
Then, it follows from the fundamental Chebshev and Markov (concentration) inequalities that 
\begin{align}\label{inequality:inspiration}
\left\{ \begin{aligned}
& \Pr\left\{  | (\Theta X^\mathsf{T})^{ii}|< c_1 \sqrt{\mathbb{D}[ (\Theta X^\mathsf{T})^{ii} ]} \right\} \ge 1-\frac{1}{c_1^2}  \\
& \Pr \left\{  | (X X^\mathsf{T})^{ii} | > c_2 \mathbb{E}[ (X X^\mathsf{T})^{ii} ] \right\}  \le \frac{1}{c_2 }
\end{aligned}\right.,
\end{align}
where $c_1,c_2\ge1$ are used to control the probability tails, and the terms $c_1 \sqrt{\mathbb{D}[ (\Theta X^\mathsf{T})^{ii} ]}$ and $c_2 \mathbb{E}[ (X X^\mathsf{T})^{ii} ]$ serve as the deviation quantification of $(\Theta X^\mathsf{T})^{ii}$ and $ (X X^\mathsf{T})^{ii}$, respectively. 
It is worth noting that although the probability tails in \eqref{inequality:inspiration} are weak to give the tail for 
\begin{align}\label{eq:prob_tail}
\Pr\left\{  \frac{| (\Theta X^\mathsf{T})^{ii}| }{| (X X^\mathsf{T})^{ii} | }< \frac{c_1 \sqrt{\mathbb{D}[ (\Theta X^\mathsf{T})^{ii} ]}}{c_2 \mathbb{E}[ (\Theta X^\mathsf{T})^{ii} ]}  \right\},
\end{align}
this inadequacy can be addressed by using more sophisticated concentration inequalities with the same deviation quantification but sharper probability tails. 
For example, one can apply the general Hoeffding inequality to $(\Theta X^\mathsf{T})^{ii}$ and the Hanson-Wright inequality to $(X X^\mathsf{T})^{ii}$ to derive the tail for \eqref{eq:prob_tail}. 
Since this point is not the focus of this paper, we refer to \cite[Chap. 2 and 6]{Roman2018High} for more details. 

In summary, \eqref{inequality:inspiration} provides a way of using the variance expectation quantities to characterize $\| \Theta X^\mathsf{T} \|$ and  $\|X X^\mathsf{T}\|$, which can be bounded by $| (\Theta X^\mathsf{T})^{ii}|$ and $| (X X^\mathsf{T})^{ii} | $, respectively. 
Meanwhile, the probability tail for the deviation quantification of the ratio in \eqref{eq:prob_tail} is tractable by existing concentration inequalities. 
Hence, $R_{\theta}(T)$ is an effective metric to characterize the decaying rate of $\|\hat{W}(T) \!-\! W\| $, and we will further prove its relationship with the well-investigated sample complexity in the following.

\begin{remark}
The reason for adopting the trace operation on matrices in $R_{\theta}(T)$ instead of norms lies in two aspects. 
First, the matrix trace has close connections with the norm (e.g., $|\operatorname{tr}( M_0)|\le n \| M_0  \|$ and $\operatorname{tr}( M_0^\mathsf{T} M_0 )=\| M_0  \|_{F}^2$), and thus some common norm inequalities are also applicable. 
Second, the trace of two matrix product is independent of the multiplication order, and the trace of multiple matrix summation is invariant under a similarity transformation. 
The two properties make it more convenient than the matrix norm to characterize the expectation and variance of a random matrix (e.g., Lemma \ref{le:equivalent_trans} in this paper). 
\end{remark}

\subsection{Relationship With The Sample Complexity}

In this part, we reveal the relationship between the proposed $R_{\theta}(T)$ and the sample complexity \eqref{eq:sample_complexity} in terms of the same decaying rate. 



First, the state evolution \eqref{eq:global_model} can be expanded as
\begin{align}\label{eq:x_expansion}
x_t=W^t x_0 + \sum_{m = 0}^{t-1} W^{t-m-1} \theta_m = x_t^* + \tilde{\theta}_{t-1},
\end{align}
where $\tilde{\theta}_{t-1}=\sum_{m = 0}^{t-1} W^{t-m-1} \theta_m$. 
Some properties for expectation and variance are also present for analyzing $R_{\theta}(T)$. 
Given $T$ groups of random variables $\{\omega_\ell\}_{\ell=1}^{T}$, the expectation and variance formulas for their weighted sum are given by 
\begin{align}
\!\!&\mathbb{E}[\sum_{\ell= 1}^{T} c_\ell\omega_\ell]= \sum_{\ell= 1}^{T} c_\ell \mathbb{E}[ \omega_\ell], \\
\!\!&\mathbb{D}[\sum_{\ell= 1}^{T} c_\ell \omega_\ell]= \sum_{\ell= 1}^{T} c_\ell^2 \mathbb{D}[ \omega_\ell] \! + \! \!\!\!\!\!\!\! \sum_{\ell_1,\ell_2=1,\ell_1\neq \ell_2}^{T} \!\!\!\!\!\! \!\!\! \operatorname{Cov}[ c_{\ell_1}\omega_{\ell_1}, c_{\ell_2}\omega_{\ell_2}]. \label{formula:sum_variance}
\end{align}
For two independent random variables $\omega_{\ell_a}$ and $\omega_{\ell_b}$, the expectation and variance of their product are given by 
\begin{align}
\!\!&\mathbb{E}[\omega_{\ell_a} \omega_{\ell_b}]= \mathbb{E}[\omega_{\ell_a}] \mathbb{E}[\omega_{\ell_b}], \label{eq:expectation_product}\\
\!\!&\mathbb{D}[\omega_{\ell_a} \omega_{\ell_b}] = \mathbb{D}[\omega_{\ell_a}] \mathbb{D}[\omega_{\ell_b}] \!+\! \mathbb{D}[\omega_{\ell_a}] \mathbb{E}^2[\omega_{\ell_b}] \!+\! \mathbb{D}[\omega_{\ell_b}] \mathbb{E}^2[\omega_{\ell_a}]. \label{eq:variance_product}
\end{align}

The following lemma presents the expectation and variance of a quadratic random variable.  
\begin{lemma}\label{le:exp_var}
Consider $\theta_t$ is time-independent zero-mean noise and $\mathbb{E}[(\theta_t^i)^4]<\infty,\forall i \in \mathcal{V}$. 
Given arbitrary $Q\in\mathbb{R}^{nT\times nT}$, 
the quadratic variable $z=\bm{\theta}_{0:T-1}^\mathsf{T} Q \bm{\theta}_{0:T-1}$ satisfies
\begin{align}
\!\!\!&\mathbb{E}[z] \!=\! \sum_{\ell= 1}^{nT} Q_{\ell\ell} \mathbb{D}[\bm{\theta}_{\ell}],  \\
\!\!\!&\mathbb{D}[z] \!=\!\! \sum_{\ell= 1}^{nT} Q_{\ell\ell}^2 (\mathbb{E}[\bm{\theta}_{\ell}^4] \!-\! \mathbb{D}^2[\bm{\theta}_{\ell}])  +  \!\!\!\! \!\!\!\!\!\! \!\!\sum_{\ell_1,\ell_2= 1, \ell_1\neq\ell_2}^{nT} \!\! \!\!\!\!\!\! \!\!\!\! Q_{\ell_1 \ell_2 }^2 \mathbb{D}[\bm{\theta}_{\ell_1}] \mathbb{D}[\bm{\theta}_{\ell_2}],
\end{align}
where $\bm{\theta}_{\ell}$ is the $\ell$-th element of $\bm{\theta}_{0:T-1}$. 
\end{lemma}

\begin{proof}
The proof is provided in Appendix \ref{pr:le:exp_var}. 
\end{proof}


Next, we give the convergence rate of $R_{\theta}(T)$ and show its relationship with that of the sample complexity about $\hat{W}(T)$. 
Note that the sample complexity analysis in existing works does not require $W$ to satisfy Assumption \ref{assu:topo}, but assumes that the noises are i.i.d. zero-mean sub-Gaussian noises. 

\begin{theorem}\label{th:equivalent_rate}
Considering that the NS is updated by \eqref{eq:global_model}, $\{\theta_t\}_{t=0}^T$ are i.i.d. sub-Gaussian noises with zero-mean and $ \mathbb{E}[\theta_t \theta_{t}^\mathsf{T}] =\sigma_0^2 I_n,~t\in\mathbb{N}$, we have 
\begin{align}\label{eq:rates_eq}
R_{\theta}(T)=\left\{
\begin{aligned}
& \bm{O}(\frac{1}{\sqrt{T}}),&&~\text{if}~\rho(W) < 1 \\
&\bm{O}(\frac{1}{T}),&&~\text{if}~\rho(W)=1 \\
&\bm{O}(\frac{1}{\rho^T(W) }),&&~\text{if}~\rho(W)>1 \\
\end{aligned}\right..
\end{align}
\end{theorem}

\begin{proof}
The proof is provided in Appendix \ref{pr:th:equivalent_rate}. 
\end{proof}

\begin{table}[t]
\caption{\label{tab:rates_comparison} The decaying rates of $R_{\theta}(T)$ and the sample complexity \eqref{eq:sample_complexity} under i.i.d. zero-mean sub-Gaussian noises}
\begin{tabular}{ccc}
\toprule[1pt]
\multicolumn{1}{c}{Spectral radius}     & \multicolumn{1}{c}{ $\bm{R_{\theta}(T)}$ }  & \makecell[c]{The sample complexity \eqref{eq:sample_complexity}$^{^1}$}    \\ \hline
\makecell[c]{ $\rho(W)<1$ } & \makecell[c]{ $\bm{O}(\frac{1}{\sqrt{T}})$ }   & \makecell[c]{$\bm{O}(\frac{1}{\sqrt{T}})$$^{^2}$  \cite{simchowitz2018learning} }  \\ \hline
\makecell[c]{ $\rho(W)=1$ } & \makecell[c]{ $\bm{O}(\frac{1}{T})$ } & \makecell[c{p{4.7cm}}]{ lower bound: $\bm{O}(\frac{1}{T})$ \cite[Theorem 2.3]{simchowitz2018learning} \\ upper bound: $\bm{O}(\frac{1}{T})$ if $W$ is diagonalizable \cite[Corollary A.3]{simchowitz2018learning}, or $\bm{O}(\frac{\log T}{T})$ otherwise \cite[Theorem 1]{sarkar2019near} } \\ \hline 
\makecell[c]{ $\rho(W)>1$  } & \makecell[c]{ $\bm{O}(\frac{1}{\rho^T })$ }  & \makecell[c]{ $\bm{O}(\frac{1}{\rho^T })$$^{^3}$ \cite{sarkar2019near} }  \\ 
\bottomrule[1pt]
\end{tabular}
\begin{tablenotes}[para]\footnotesize
    \item[1] The decaying rate of \eqref{eq:sample_complexity} is the rate of how $\varepsilon(\delta,T)$ will decay to zero with $T$ growing, given an arbitrary small $\delta$. \\
    \item[2] $\bm{O}(\frac{1}{\sqrt{T}})$ applies to both the upper and lower bounds of the sample complexity (see Theorem 2.1 and 2.3 in \cite{simchowitz2018learning}, respectively). \\
    \item[3] $\bm{O}(\frac{1}{\rho^T })$ applies to both the upper and lower bounds of the sample complexity (see Theorem 1 and Proposition 4.1 in \cite{sarkar2019near}, respectively). 
\end{tablenotes}
\end{table}

Theorem \ref{th:equivalent_rate} gives the decaying rate of $R_{\theta}(T)$ under the same conditions that are used in the sample complexity analysis of $\|\hat{W}(T)\!-\!W\|$ in the literature. 
Specifically, Table \ref{tab:rates_comparison} provides the detailed comparisons to reveal that $R_{\theta}(T)$ is a well-defined quantity and coincides with the sample complexity in terms of the decaying rate. 
It is worth noting that when $\rho(W)=1$, the decaying rate of $R_{\theta}(T)$ may differ up to a logarithmic factor from that of the upper bound of the sample complexity, which exhibits as $\bm{O}(\frac{\log T}{T})$ in some situations. 
In fact, if $W$ is a diagonalizable matrix, the above rate will reduce to $\bm{O}(\frac{1}{T})$ \cite[Corollary A.3]{simchowitz2018learning}. 
We observe that in most settings of $W$, the decaying rate when $\rho(W)=1$ exhibits almost the same as $\bm{O}(\frac{1}{T})$, considering $\log T$ increases extremely slow (e.g., $\log 10^{100}=100\log 10$). 
Therefore, $R_{\theta}(T)$ is sufficient to characterize the decaying rate of $\|\hat{W}(T)-W\|$.


Based on the simplicity and effectiveness of $R_{\theta}(T)$, we further leverage it to characterize the performance of topology preservation performance and obtain the optimal decaying-variance design in the following sections. 
Note that the definition \eqref{eq:def_trace} filters the influence of the initial state $x_0$, and for simple expressions, we directly consider $x_0=0$ and use the following metric hereafter
\begin{itemize}
\item \textbf{The simplified version of $R_{\theta}(T)$ defined by \eqref{eq:def_trace}}:
\begin{equation}\label{eq:def_trace2}
R_{\theta}(T)=\frac{ \sqrt{ \operatorname{tr}(\mathbb{D}\left[\Theta X^\mathsf{T}  \right] )}  }{ \operatorname{tr}(\mathbb{E}[X X^\mathsf{T}] ) }~~(x_0=0),
\end{equation}
\end{itemize}
which is proved equivalent with \eqref{eq:def_trace} in Appendix \ref{pr:th:equivalent_rate}.

\section{Topology Preservation by Independent Noises}\label{sec:independent}

In this section, we first give the expressions of the independent noises. 
Then, we use the proposed metric to analyze the topology preservation performance under the designed noises.

\subsection{Constraining The State Deviation by Decaying Noises}

First, recall that Theorem \ref{th:equivalent_rate} reveals that the topology matrix can be accurately inferred in the asymptotic sense, but the convergence of the state deviation $\| x_t - x^*_t \|$ is not guaranteed. 
In fact, if the added noises are in i.i.d. forms, then it is straightforward from \eqref{eq:x_expansion} to obtain that 
\begin{align}
\mathbb{E}[\| x_t - x^*_t \|^2] = & \operatorname{tr}( \mathbb{E}[ \tilde{\theta}_{t-1} \tilde{\theta}_{t-1}^\mathsf{T})]) = \sigma_0^2 \operatorname{tr}(\Gamma_t),
\end{align}
which is divergent in the mean square sense when $\rho(W)=1$. 
To avoid the explosive state deviation issue, we highlight that it is necessary for the added noises to meet the following variance-decaying condition
\begin{equation}\label{eq:necessary_decaying}
 \mathop {\lim }\limits_{t \to \infty} \sigma_t^2 =0. 
\end{equation} 
The necessity of \eqref{eq:necessary_decaying} is easy to verify by contradiction. 
Suppose that there exists a variance sequence $\{\sigma_t^2\}$ satisfying $\mathop {\lim }\limits_{t \to \infty} \sigma_t^2=C_{\sigma}>0$ such that $\mathop {\lim }\limits_{t \to \infty}  \mathbb{E}[\| x_t - x^*_t \|^2]=0$. 
Then, there exists a group of $C_{\sigma_1}>0$ and $C_{\sigma_2}>0$ such that $ t C_{\sigma_2} \le \mathbb{E}[\| x_t - x^*_t \|^2]\le t C_{\sigma_2}  $. 
Let $t\to\infty$ and the contradiction with \eqref{eq:necessary_decaying} occurs.

Motivated by the necessary decaying requirement \eqref{eq:necessary_decaying}, we consider that the variance of added $\theta_t$ decays in a polynomial form, given by 
\begin{equation}\label{eq:noise_variance}
\sigma_t^2=\frac{\sigma_0^2}{{(t+1)}^\alpha},~t\in\mathbb{N},
\end{equation} 
where $\alpha\ge0$ is the constant power to control the decaying rate of $\sigma_t^2$. 
Specifically, the case of $\alpha=0$ corresponds to that the added noises are i.i.d. noises. 
We also note that adding noises of exponentially-decaying forms has been used to preserve the state privacy of NSs with undirected topologies, e.g., \cite{mo2017privacy,he2018privacy}. 
Nevertheless, how the decaying noises will affect the preservation performance for general directed topologies is still lacking analysis.

\begin{remark}
The reasonability for considering \eqref{eq:noise_variance} is two-fold. 
First, most decaying function forms can be approximated by multiple groups of polynomial functions. 
Specifically, the decaying rate in \eqref{eq:noise_variance} is modest compared with the commonly-used exponentially-decaying rate, and the obtained results for \eqref{eq:noise_variance} will easily accommodate the latter case. 
Second, the summation $\sum_{t = 0}^{T-1} \frac{1}{{(t+1)}^\alpha}$ exhibits various convergence/divergence behaviors for different $\alpha$. 
The two points help us to exploit the fine-grained relationship between the preservation performance and various decaying noises. 
\end{remark}

In the sequel of this paper, we consider that $\{\theta_t\}_{t=0}^{T-1}$ subject to the following assumption. 
\begin{assumption}\label{assu:independent_noise}
The added independent noises $\{\theta_t\}_{t=0}^{T-1}$ satisfy
\begin{align}\label{eq:noise_requirement}
\left\{\begin{aligned}
&\mathbb{E}[\theta_t^i]=0,~\mathbb{E}[(\theta_t^i)^4]<\infty,\forall i \in \mathcal{V} \\
& \mathbb{E}[\theta_t \theta_t^\mathsf{T}]=\frac{\sigma_0^2}{{(t+1)}^\alpha}I,~\mathbb{E}[\theta_{t_1} \theta_{t_2}^\mathsf{T}]=0 (t_1\neq t_2)
\end{aligned}\right..
\end{align}
\end{assumption}


\subsection{Topology Preservation Analysis Under Independent Noises}

First, note that when $\alpha\neq 0$, the conclusions in Theorem \ref{th:equivalent_rate} are not applicable because the noises are no longer identically distributed. 
Since the noise accumulation plays a key role in the rate characterization for $R_{\theta}(T)$, 
we define the following variance sum for $\{\sigma_t^2\}_{t=0}^{T-1}$ to facilitate the analysis  
\begin{equation}
h_{\alpha}(T)=\sum_{t = 0}^{T-1}\frac{\sigma_0^2}{(t+1)^{\alpha}}. 
\end{equation}
Accordingly, the increment scale of $h_{\alpha}(T)$ is demonstrated by the following lemma. 



\begin{lemma}\label{le:approximate}
When $T$ grows, the increment scale of $ h_{\alpha}(T)=\sum_{t=1}^{T}\frac{\sigma_0^2}{(t+1)^{\alpha}} $ is characterized by
\begin{equation}
h_{\alpha}(T)=\bm{O}( F_{\alpha}(T) ), 
\end{equation}
where $F_{\alpha}(T)=\int_1^T \frac{\sigma_0^2}{y^\alpha} d y$. 
\end{lemma}

\begin{proof}
The proof is provided in Appendix \ref{pr:le:approximate}. 
\end{proof}

The key insight of Lemma \ref{le:approximate} lies in using Euler-Maclaurin summation formula \cite{apostol1999elementary} to approximate $h_{\alpha}(T)$. 
It is intuitive to think that the smaller $\alpha$ is, the more noisy the observations will be, which might benefit to preserve the topology. 
However, the following result will reveal that the above intuition is not always true and give the optimal $\alpha$ to hinder the topology inference performance.

\begin{theorem}\label{th:decaying_rate}
Considering that the NS is updated by \eqref{eq:global_model} and Assumption \ref{assu:topo}-\ref{assu:independent_noise} hold,  
the decaying rate of $R_{\theta}(T)$ is given by  
\begin{align}\label{eq:limit_Re}
\!\! R_{\theta}(T)  = \left\{\begin{aligned}
& \sqrt{  \bm{O}(\frac{1}{ T^2}) \!+\!  \bm{O}(\frac{1}{T^{3\!-\!\alpha}})  },&&\text{if}~\alpha\!<\!1 \\
&   \sqrt{ \bm{O}(\frac{1}{T^2}) \!+\! \bm{O}( \frac{1}{ T^2 \log T})} , &&\text{if}~\alpha\!=\!1 \\
&   \bm{O}(\sqrt{\frac{\alpha-1}{2^\alpha}} \frac{1}{T}) ,&&\text{if}~\alpha \!>\!1
\end{aligned} 
\right..
\end{align}
\end{theorem}

\begin{proof}
The proof is provided in Appendix \ref{pr:th:decaying_rate}. 
\end{proof}

Theorem \ref{th:decaying_rate} explicitly characterizes the decaying tendency of the inference error, which always decays to zero as $T$ goes to infinity. 
Specifically, the rate characterization well accommodates the case of $\rho(W)=1$ in Theorem \ref{th:equivalent_rate} by letting $\alpha=0$, and the noise-free case by letting $\alpha\to\infty$. 
It is worth noting from Theorem \ref{th:decaying_rate} that, the non-asymptotic topology preservation performance under decaying noises does not necessarily have monotonicity about the decaying rate of the noise variance. 
The following theorem presents the optimal $\alpha$ for the preservation performance characterized by $R_{\theta}(T)$. 
\begin{theorem}\label{th:optimal}
Under the same preconditions of Theorem \ref{th:decaying_rate}, the optimal $\alpha$ with the slowest decaying rate for $R_{\theta}(T)$ is
\begin{equation}\label{eq:optimal_alpha}
\alpha^*=\frac{1+\log 2}{\log 2}.
\end{equation}
\end{theorem}
\begin{proof}
This result is straightforward from Theorem \ref{th:decaying_rate}. 
When $\alpha\le1$, it is clear to see that a larger $\alpha$ will incur a slower decaying tendency. 
Thus, we turn to focus on the situation when $\alpha>1$. 
To describe the monotonicity of $\bm{O}(R_{\theta}(T))$ about $\alpha$, we define $C_\alpha=\frac{\alpha-1}{2^\alpha}$. 
It is straightforward from the derivative $\frac{\mathrm{d} C_\alpha}{\mathrm{d} \alpha}$ to obtain that $C_\alpha$ will increase first and then decrease with $\alpha$ growing. 
When $C_\alpha$ reaches the maximum value, the corresponding $\alpha$ is obtained by solving 
\begin{align}
\frac{\mathrm{d} C_\alpha}{\mathrm{d} \alpha}=  \frac{1-(\alpha-1)\log 2}{2^\alpha}=0,
\end{align}
which gives the optimal $\alpha^*$ for preserving the topology. 
\end{proof}

Based on the above analysis, we further provide the following conditions of using more general time-independent decaying noises for topology preservation and characterize the corresponding decaying rate of $R_{\theta}(T)$. 
\begin{corollary}\label{coro:sufficient_noise}
Consider that the NS is updated by \eqref{eq:global_model}, Assumption \ref{assu:topo} holds, and the variance of $\theta_t$ is $\sigma_t^2=\sigma_0 g(t)$, where $g(t)$ is a general non-increasing continuous function about $t$. 
If $g(t)$ satisfies the following calculus condition
\begin{equation}\label{eq:converge_condition}
\mathop {\lim }\limits_{T \to \infty} \frac{g^{(2T)}(2T) - g^{(2T)}(1) }{(2 \pi)^{2T}} =0,
\end{equation}
where $g^{(2T)}(2T)$ represents the $2T$-order derivative at point $2T$, 
then the decaying rate of $R_{\theta}(T)$ can be characterized by
\begin{align}\label{eq:generate_rate}
R_{\theta}(T)=\bm{O}\left( \frac{ \sqrt{ \int_1^{T\!-\!1}  g(y_2) (\int_0^{y_2-1} g(y_1)  d{y_1} ) d{y_2} }  }{  T \int_0^{T-2} g(y) d{y} } \right).
\end{align}
\end{corollary}

\begin{proof}
The proof is provided in Appendix \ref{pr:coro:sufficient_noise}. 
\end{proof}

Corollary \ref{coro:sufficient_noise} provides an efficient way to characterize the decaying rate of $R_{\theta}(T)$ under a broader class of time-independent noises, whose variances are non-increasing with time. 
By adding those noises, the error $\|\hat{W}(T)-W\|$ will always decay to zero asymptotically in the sense of variance-expectation ratio. 
For example, if the noise variance decays exponentially in the form of $\sigma_t^2=\sigma_0^2 q^t $ ($q\in(0,1)$), then it is easy to verify that $g(t)$ satisfies the condition \eqref{eq:converge_condition} (here the fact $\mathop {\lim }\limits_{T \to \infty} g^{(T)}(T)=\mathop {\lim }\limits_{T \to \infty}(q\log{q})^T=0$ is applied), and a smaller $q$ incurs a faster decaying rate of $R_{\theta}(T)$.

\section{Topology Preservation by Dependent Noises}\label{sec:dependent}

In this section, we extend the topology preservation analysis to the case of adding dependent decaying noises, and characterize the tradeoff between the state deviation and the topology preservation.

\subsection{Decaying Noises With One-lag Time Dependence}

Although the decaying-noise condition \eqref{eq:necessary_decaying} is necessary to constrain the state deviation while preserving the topology, adding such noises still cannot ensure the zero state deviation requirement \eqref{eq:optim-c}, which is desirable in many situations. 
Taking $\sigma_t^2=\frac{\sigma_0^2}{{(t+1)^\alpha}}$ as an example, the state deviation is given by 
\begin{align}\label{eq:state_divergence}
\mathbb{E}[\| x_t - x^*_t \|^2] = & \operatorname{tr}( \mathbb{E}[ \tilde{\theta}_{t-1} \tilde{\theta}_{t-1}^\mathsf{T})]) \nonumber \\
=& \sum_{m = 0}^{t-1} \frac{\sigma_0^2 \operatorname{tr}( W^{t-m-1} (W^{t-m-1})^\mathsf{T}) }{{(m+1)^\alpha}},
\end{align}
which will even diverge when $\alpha\le1$. 
To further overcome the divergent or bounded state deviation issue, we introduce the noise dependence between noises at different moments.


Note that the time-dependent noise vectors can always be represented as a linear combination of a group of time-independent noise vectors. 
Based on the independent noises $\{\theta_t\}_{t=0}^{T-1}$ and for simplicity, we first consider the added noises satisfying the following one-lag time dependence assumption. 
Extensions to the noises with multi-lag time dependence will be provided at the end of this section. 
\begin{assumption}\label{assu:dependent_noise}
The dependent noises $\{\xi_t\}_{t=0}^{T-1}$ are given by 
\begin{equation}\label{eq:cor_noise}
\xi_t=\theta_t - \theta_{t-1},~t\in\mathbb{N}^{+},
\end{equation}
and $\xi_0=\theta_0$, where $\{\theta_t\}_{t=0}^{T-1}$ subject to Assumption \ref{assu:independent_noise}. 
\end{assumption}


Note that the one-lag time dependence of $\{\xi_t\}_{t=0}^{T-1}$ is described by $\mathbb{E}[\xi_t \xi_{t-1}^\mathsf{T}]=-\frac{\sigma_0^2}{t^\alpha}$. 
The benefit of adopting $\xi_t$ lies in two aspects. 
First, the sum of $\{\xi_t\}_{t=0}^{T-1}$ is given by $\sum_{t = 0}^{T-1} \xi_t = \theta_{T-1}$, which implies that the noise accumulation $\sum_{t = 0}^{T-1} \xi_t$ has the same decaying rate as that of $\theta_{T-1}$. 
Second, from the storage perspectives, \eqref{eq:cor_noise} only requires each node in the network to store and update its own historical and current added noises, which relies on no global information about $W$ and is convenient to implement in practice.

To avoid notation confusions with the case of adding independent noise $\theta_t$, the state evolution when adding dependent $\xi_t$ is written as 
\begin{equation}\label{eq:state_eta}
x_{\xi,t}=Wx_{\xi,t-1}+\xi_{t-1},
\end{equation}
and the corresponding OLS-based topology estimator and the inference error are given by  
\begin{align}
&\hat{W}_{\xi}= X_{\xi}^+ X_{\xi}^\mathsf{T} (X_{\xi} X_{\xi}^\mathsf{T})^{-1}, \label{eq:topo_estimator2} \\
&\|\hat{W}_{\xi}(T)-W\|= \| \Xi X_{\xi}^\mathsf{T} (X_{\xi} X_{\xi}^\mathsf{T})^{-1} \|, \label{eq:topo_error2}
\end{align}
where $X_{\xi}\!=\![x_{\xi,0},x_{\xi,2},\cdots,x_{\xi,T\!-\!1}]$, $\Xi \!=\! [\xi_0,\xi_1,\!\cdots\!,\xi_{T\!-\!1}]$, and $X_{\xi}^+ \!=\![x_{\xi,1},x_{\xi,2},\!\cdots\!,x_{\xi,T}]$.


\subsection{State Deviation Analysis Under Dependent Noises}

In this part, we show how the correlated noise design \eqref{eq:cor_noise} will amend the state deviation issue. 

First, the state deviation by adding $\xi_t$ at time $t$ is given by 
\begin{align}\label{eq:x_xi}
x_{\xi,t} - x^*_t= &  \sum\nolimits_{m = 1}^{t-1} W^{t-m-1} (\theta_{m}-\theta_{m-1})  + W^{t-1}\theta_{0} \nonumber \\
=&  \theta_{t-1} \!+\!  \sum\nolimits_{m = 0}^{t-2} (W^{t-m-1} - W^{t-m-2})\theta_m  \nonumber \\
=& \theta_{t-1} +  \sum\nolimits_{m = 0}^{t-2} \tilde{W}_{t,m} \theta_m ~~(t\ge 2),
\end{align}
where $\tilde{W}_{t,m}= W^{t-m-1} - W^{t-m-2}$, $x_{\xi,0} - x^*_0 =0$ and $x_{\xi,1} - x^*_1=\theta_0$. 
Then, the following result reveals how the state deviation evolves with $t$. 
\begin{theorem}\label{th:state_bound}
Consider that the NS is updated by \eqref{eq:state_eta}, and Assumption \ref{assu:topo} and \ref{assu:dependent_noise} hold. 
If the decaying parameter $\alpha\!>\!0$, then the state deviation of \eqref{eq:state_eta} satisfies
\begin{align}\label{eq:zero_convergence}
\mathop {\lim }\limits_{t \to \infty} \mathbb{E}[\|x_{\xi,t} - x^*_t \|^2] = 0. 
\end{align}
Specifically, the larger $\alpha$ is, the faster $\mathbb{E}[\|x_{\xi,t} - x^*_t \|^2]$ will converge to zero. 
\end{theorem}

\begin{proof}
The proof is provided in Appendix \ref{pr:th:state_bound}. 
\end{proof}

Theorem \ref{th:state_bound} shows that by leveraging the one-lag dependent noise design \eqref{eq:cor_noise} with $\alpha>0$, the state deviation with the ideal situation will decay to zero asymptotically in the mean square sense. 
Compared with the divergent or bounded state deviation when adding $\theta_t$, the zero asymptotic state deviation \eqref{eq:zero_convergence} achieved by added $\xi_t$ has no dependence on the increment scale of $\sum\nolimits_{t = 0}^{T-1} \sigma_t^2 $. 
We remark that this convergence improvement is brought by the intrinsic exponential convergence of the matrix $\tilde{W}_{t,m}$ and the decaying nature of $\xi_t$. 
Even if the variance is not decaying (i.e., $\alpha\neq0$), the deviation $\mathop {\lim }\limits_{t \to \infty} \mathbb{E}[\|x_{\xi,t} \!-\! x^*_t \|^2] $ is bounded, which is still better than the divergent $\mathop {\lim }\limits_{t \to \infty} \mathbb{E}[\|x \!-\! x^*_t \|^2] $ when adding $\theta_t$ with $\alpha=0$.

\subsection{Topology Preservation Analysis Under Dependent Noises}

In this part, we characterize the topology preservation performance where $\xi_t$ is added at each iteration. 

First, it should be pointed out that due to the dependence between $\xi_{t}$ and $\xi_{t-1}$, the correlation between $\Xi$ and $X_{\xi}^\mathsf{T}$ is no longer zero (i.e., $\mathbb{E}[\Xi X_{\xi}^\mathsf{T}]\neq0$). 
Under this circumstance, indirectly using $R_{\theta}(T)$ to characterize the decaying rate of $\|\hat{W}_{\xi}(T)-W\|$ is inappropriate. 
Similar to \eqref{inequality:inspiration}, applying the Chebyshev inequality on the element $\Xi^{i} X_{\xi}^{i}$, we have 
\begin{align}\label{eq:cheb2}
\Pr\left\{  \left | \Xi^{i} X_{\xi}^{i} - \mathbb{E}[\Xi^{i} X_{\xi}^{i}]  \right|< c_\sigma \sqrt{\mathbb{D}[\Xi^{i} X_{\xi}^{i}]} \right\}  \ge 1-\frac{1}{c_\sigma^2},  
\end{align}
where $1<c_\sigma<\infty$ is a preset parameter. 
Apparently, \eqref{eq:cheb2} further implies that $\Xi^{i} X_{\xi}^{i}$ locates in the following interval with high probability
\begin{equation}\label{eq:random_deviation2}
\Xi^{i} X_{\xi}^{i}\in \left (\mathbb{E}[\Xi^{i} X_{\xi}^{i}] \!-\! c_\sigma \sqrt{\mathbb{D}[\Xi^{i} X_{\xi}^{i}]}, \mathbb{E}[\Xi^{i} X_{\xi}^{i}] \!+\! c_\sigma \sqrt{\mathbb{D}[\Xi^{i} X_{\xi}^{i}]} \right) .
\end{equation}
Inspired by \eqref{eq:random_deviation2}, we resemble the deviation quantification therein and further propose the following rate metric to accommodate $R(\| \hat{W}(T) - W \|) $ under dependent noises:
\begin{itemize}
\item \textbf{The expectation-expectation ratio interval under dependent noises}:
\begin{align}\label{eq:new_xi}
\!\!R_{\xi}(T) \!=\frac{\operatorname{tr}(\mathbb{E}[\Xi X_{\xi}^\mathsf{T}]) \pm c_\sigma \sqrt{\operatorname{tr}( \mathbb{D}[ \Xi X_{\xi}^\mathsf{T}])}}{\operatorname{tr}(\mathbb{E}[X_{\xi} X_{\xi}^\mathsf{T}])}~~(x_0\!=\!0). 
\end{align}
\end{itemize}
Slightly different from $R_{\theta}(T)$ which is a single point value, $R_{\xi}(T)$ represents a bound interval that is centered at $\frac{\operatorname{tr}(\mathbb{E}[\Xi X_{\xi}^\mathsf{T}]) }{\operatorname{tr}(\mathbb{E}[X_{\xi} X_{\xi}^\mathsf{T}])}$ and bounded by $\frac{c_\sigma \sqrt{\operatorname{tr}( \mathbb{D}[ \Xi X_{\xi}^\mathsf{T}])}}{\operatorname{tr}(\mathbb{E}[X_{\xi} X_{\xi}^\mathsf{T}])}$, which is determined by the random nature of $\Xi X_{\xi}^\mathsf{T}$ when adding $\xi_t$. 
In the sequel, we will demonstrate that the increment scale of $\frac{c_\sigma \sqrt{\operatorname{tr}( \mathbb{D}[ \Xi X_{\xi}^\mathsf{T}])}}{\operatorname{tr}(\mathbb{E}[X_{\xi} X_{\xi}^\mathsf{T}])}$ will not exceed $\frac{\operatorname{tr}(\mathbb{E}[\Xi X_{\xi}^\mathsf{T}]) }{\operatorname{tr}(\mathbb{E}[X_{\xi} X_{\xi}^\mathsf{T}])}$, and $R_{\xi}(T)$ is capable of characterizing the decaying rate of $\|\hat{W}_{\xi}(T)-W\|$.

\begin{figure*}[t]
\centering
\subfigure[$\alpha\in\{0,0.5,1,1.5,2\}$]{\label{fig:independ_state_0}
\includegraphics[width=0.26\textwidth]{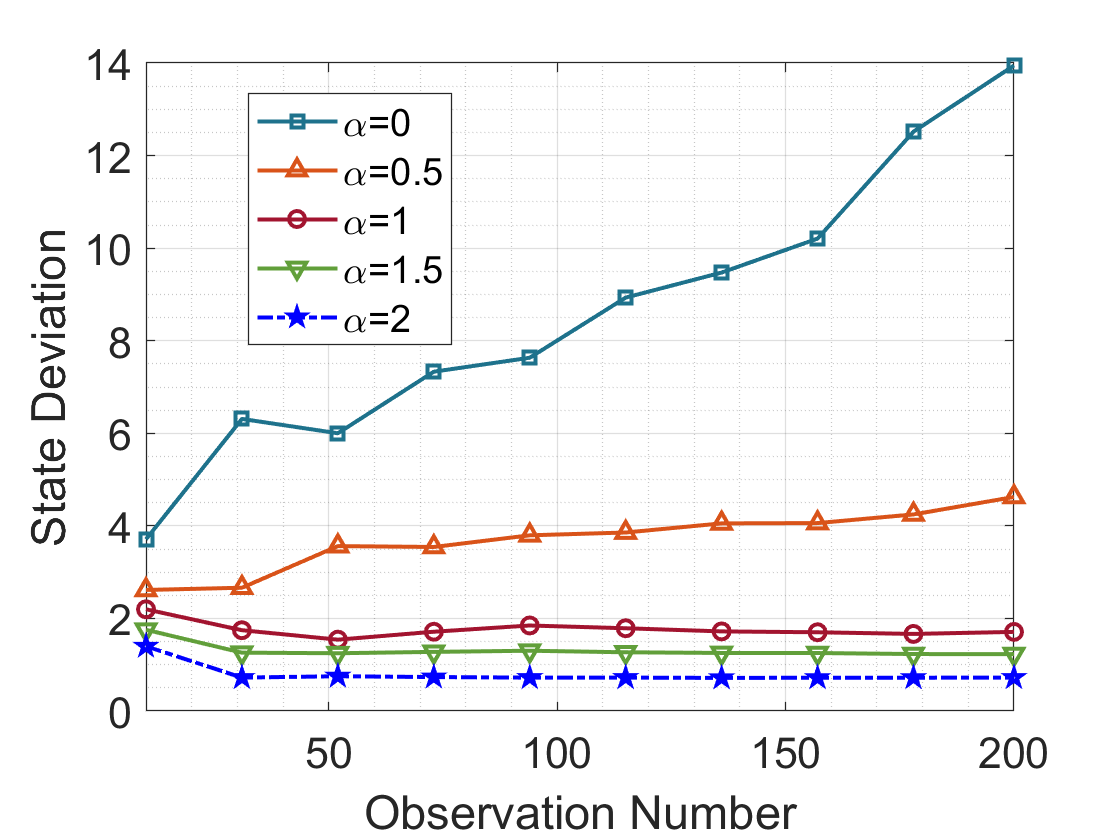}}
\hspace{-0.65cm}
\subfigure[$\alpha\in\{2.443,4,6,8,10\}$]{\label{fig:independ_state_2}
\includegraphics[width=0.26\textwidth]{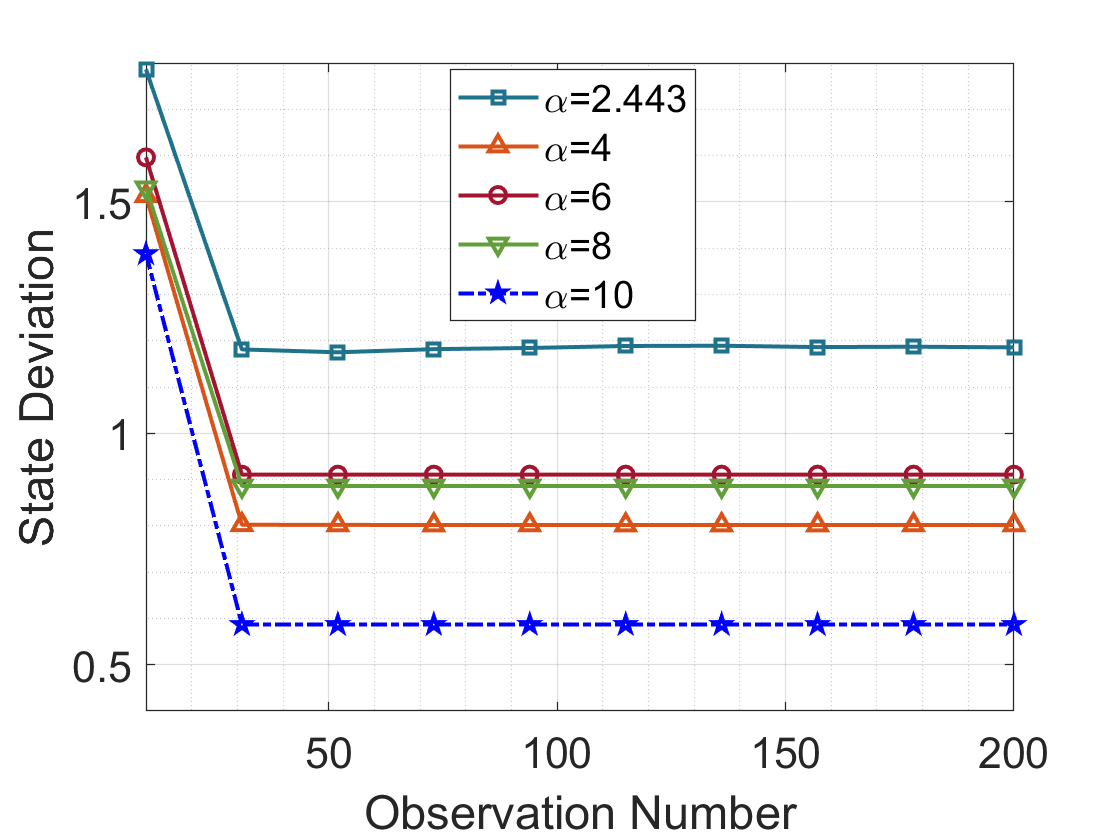}}
\hspace{-0.65cm}
\subfigure[$\alpha\in\{0,0.5,1,1.5,2\}$]{\label{fig:independ_topo_0}
\includegraphics[width=0.26\textwidth]{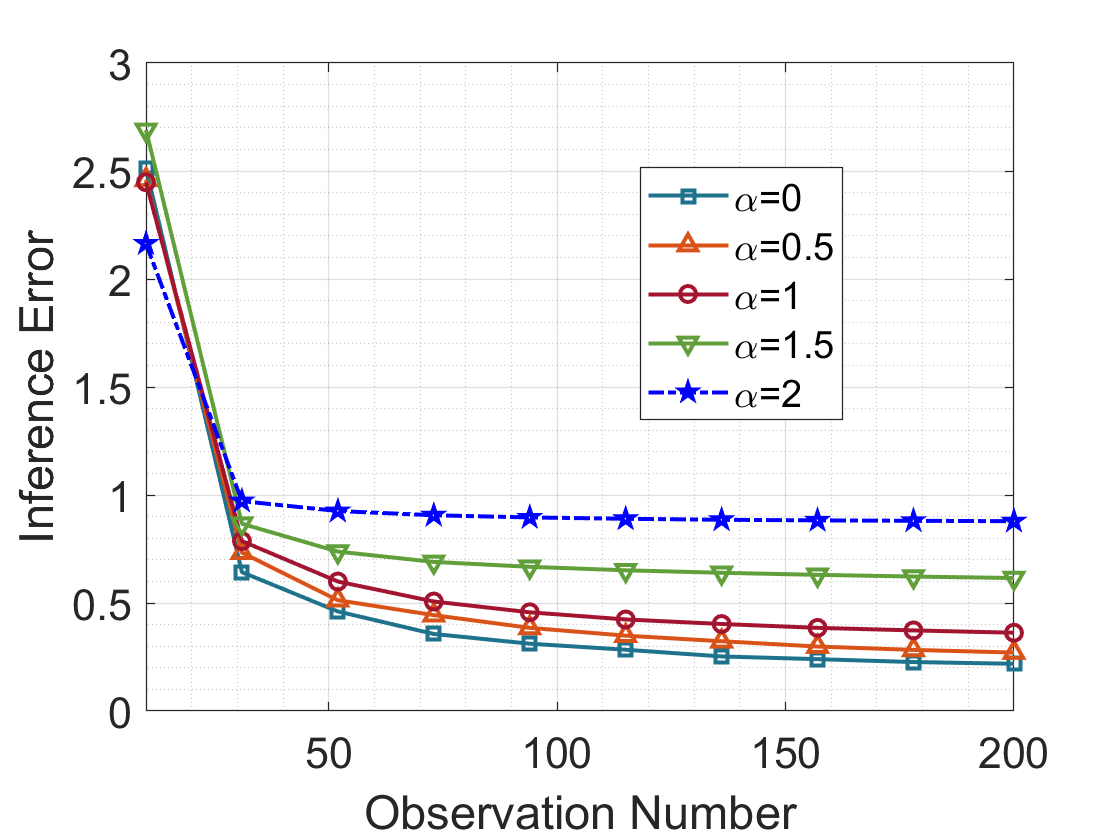}}
\hspace{-0.65cm}
\subfigure[$\alpha\in\{2.443,4,6,8,10\}$]{\label{fig:independ_topo_2}
\includegraphics[width=0.26\textwidth]{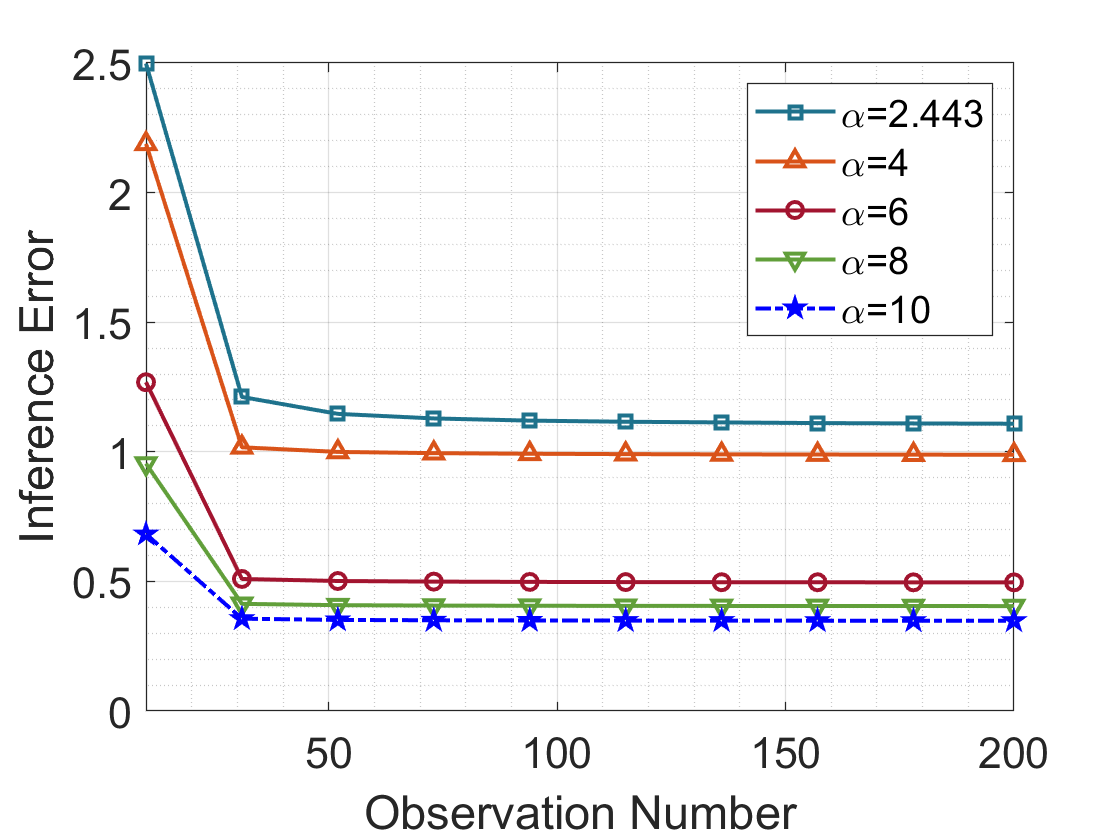}}
\vspace{-5pt}
\caption{The state deviation and topology inference error regarding $T$ when independent noises $\{\theta_t\}$ are added. 
(a)-(b): The state deviation $\|x_{t} - x^*_t \|^2$. (c)-(d): The topology inference error $\|\hat{W}(T)-W\| $. }
\label{fig:independent}
\vspace*{-13pt}
\end{figure*}

Before analyzing the decaying rate of $R_{\xi}(T)$, we observe that unlike the case of $\operatorname{tr}(\mathbb{D}[\Theta X^\mathsf{T}]) $, 
the correlation between $\xi_{t_1} x_{\xi,t_1}^\mathsf{T}$ and $\xi_{t_2} x_{\xi,t_2}^\mathsf{T}$ ($t_1\neq t_2$) makes it hard to obtain $\mathbb{D}[\Xi X_{\xi}^\mathsf{T}]$ directly. 
To overcome this analytical issue, we need the following equivalent form of $\operatorname{tr}( \mathbb{D}[ \Xi X_{\xi}^\mathsf{T}])$. 
\begin{lemma}\label{le:equivalent_trans}
Given $X_{\xi}$ and $\Xi$ defined in \eqref{eq:topo_error2}, we have 
\begin{equation}\label{eq:equivalent_transform}
\operatorname{tr}( \mathbb{D}[ \Xi X_{\xi}^\mathsf{T}])= \mathbb{D}[ \operatorname{tr}( \Xi X_{\xi}^\mathsf{T})]. 
\end{equation}
\end{lemma}
\begin{proof}
First, considering the LHS of \eqref{eq:equivalent_transform}, we have 
\begin{align}
\operatorname{tr}( \mathbb{D}[ \Xi X_{\xi}^\mathsf{T}]) =   \sum\limits_{i = 1}^{n} \mathbb{D}\left [  \sum\limits_{t = 0}^{T-1} \xi_t^{i} x_{\xi,t}^{i} \right] = \sum\limits_{i = 1}^{n} \mathbb{D} [ r(i,T)], 
\end{align}
where $r(i,T)=\sum\limits_{t = 0}^{T-1} \xi_t^{i} x_{\xi,t}^{i}$. 
As for the RHS of \eqref{eq:equivalent_transform}, it can be expanded as 
\begin{align}
\mathbb{D}[ \operatorname{tr}( \Xi X_{\xi}^\mathsf{T})]= \mathbb{D} \left[ \sum\limits_{i = 1}^{n} \sum\limits_{t = 0}^{T-1} \xi_t^{i} x_{\xi,t}^{i} \right]= \mathbb{D} \left[ \sum\limits_{i = 1}^{n} r(i,T) \right]. 
\end{align}
Although $r(i,T)$ itself contains correlated terms that are determined by $t$, the added noise $\xi_t^i$ by each node is independent of each other. 
It follows from $\mathbb{E}[\xi_{t_1}^i \xi_{t_2}^j]=0~(i\neq j)$ that 
\begin{align}\label{eq:zero_vaaa}
\!\!\operatorname{Cov}[ r(i,T), r(j,T)] \!=\! \operatorname{Cov} \left[ \sum\limits_{t = 0}^{T-1} \! \xi_t^{i} x_{\xi,t}^{i}, \sum\limits_{t = 0}^{T-1} \! \xi_t^{j} x_{\xi,t}^{j} \right] \! \!=\! 0.
\end{align}
Based on the sum-variance formula \eqref{formula:sum_variance} and \eqref{eq:zero_vaaa}, we have
\begin{align}
\mathbb{D} \left[ \sum\limits_{i = 1}^{n} r(i,T) \right]=\sum\limits_{i = 1}^{n} \mathbb{D} [ r(i,T)]=\operatorname{tr}( \mathbb{D}[ \Xi X_{\xi}^\mathsf{T}]),
\end{align}
which completes the proof. 
\end{proof}

Note that the key to the equivalence \eqref{eq:equivalent_transform} lies in that each node adds the noises independently, and whether the added noises are isotropic will not affect this equivalence. 
If the added noises have correlations between the nodes, \eqref{eq:equivalent_transform} will not necessarily hold and more fussy analysis is required to obtain $\mathbb{D}[ \Xi X_{\xi}^\mathsf{T}]$, which is left for future investigation.

Next, based on Lemma \ref{le:equivalent_trans}, the rate characterization of $R_{\xi}(T)$ is given in the following result. 

\begin{theorem}\label{th:rate_correlation}
Consider that the NS is updated by \eqref{eq:state_eta}, and Assumption \ref{assu:topo} and \ref{assu:dependent_noise} hold. 
The decaying rate of $R_{\xi}(T)$ is
\begin{align}\label{eq:rates_cor}
\!\! R_{\xi}(T)\!= \!\left\{\begin{aligned}
& \bm{O}(C)\! \pm c_\sigma \!\bm{O}(\frac{1}{\sqrt{T^{1-\alpha}}}), && \text{if}~\alpha\!<\!1 \\
& \bm{O}(C)\! \pm \! c_\sigma \bm{O}(\frac{1}{\sqrt{\log T}} ), &&\text{if}~\alpha\!= \!1 \\
&\bm{O}(C) \! \pm \! c_\sigma  \bm{O}(  \sqrt{\frac{\alpha-1}{2^\alpha(1-\frac{1}{T^{\alpha-1}})}} ), &&\text{if}~\alpha\!>\!1 \\
\end{aligned}\right. ,
\end{align}
where $C>0$ is a bounded constant irrelevant to $T$. 
\end{theorem}

\begin{proof}
The proof is provided in Appendix \ref{pr:th:rate_correlation}. 
\end{proof}

It is straightforward from Theorem \ref{th:rate_correlation} to obtain that by using \eqref{eq:cor_noise}, $\|\hat{W}_{\xi}(T)-W\|$ will not decay to zero but a constant regardless of $T$ in the sense of the expectation-expectation ratio interval. 
Specifically, when $\alpha$ is determined, the decaying rate of the inference error approaching the constant is generally slower than that of using $\theta_t$. 
We observe that this is because the accumulation of $\xi_t$ is constrained by the exponential convergence of the NS itself. 

\begin{remark}
Notice that since $R_{\xi}(T)$ represents a ratio interval, there exists no explicit monotonicity in $\bm{O}(R_{\xi}(T)) $ as $T$ grows when adding $\xi_t$, which will be verified in the simulations. 
Another notable point is that the setting of $\alpha$ also will not determine the monotonicity of $\bm{O}(R_{\xi}(T)) $ because $\|\hat{W}_{\xi}(T)-W\|$ locates in an interval in the sense of the expectation-expectation ratio. 
Hence, there is no optimal $\alpha$ for the maximum $\|\hat{W}_{\xi}(T)-W\|$ or the slowest rate of $\bm{O}(\|\hat{W}_{\xi}(T)-W\|)$ to preserve the topology. 
However, this interval bound will decay to zero as $T\to\infty$ for $\alpha\le1$, and also decay as $\alpha\to\infty$. 
\end{remark}

Finally, we generalize the conclusions of Theorem \ref{th:state_bound} and \ref{th:rate_correlation} to a broader class of decaying noises with $k$-lag ($k\in\mathbb{N}^+$) time dependence. 
\begin{corollary}\label{coro:multi}
Consider that the NS is updated by \eqref{eq:state_eta}, and Assumption \ref{assu:topo} and \ref{assu:independent_noise} hold. 
The statements \eqref{eq:zero_convergence} and \eqref{eq:rates_cor} still hold if $\{\xi_t\}_{t=0}^{T-1}$ satisfy the following $k$-lag time dependence, 
\begin{align}\label{eq:new_xi2}
\xi_t= \left\{\begin{aligned}
&\sum\nolimits_{\ell = 1}^{t+1} p_{\ell} \theta_{k-\ell}, &&\text{if}~t\le k \\
&\sum\nolimits_{\ell = 1}^{k+1} p_{\ell} \theta_{t-\ell}, &&\text{if}~t> k
\end{aligned}
\right.,
\end{align}
where $\{p_\ell\}_{\ell=1}^{k}$ satisfy 
\begin{align}\label{eq:condition_p}
\sum\nolimits_{\ell = 1}^{k} p_\ell =0,~|p_\ell|\le\bar{p}<\infty, \forall \ell\in\{1,\cdots,k\}.
\end{align}
\end{corollary}

\begin{proof}
The proof is provided in Appendix \ref{pr:coro:multi}. 
\end{proof}




\section{Simulations}\label{sec:simulation}
In this section, extensive numerical simulations are provided to verify the theoretical results. 

For the simulation setting, we randomly generate a directed graph $\mathcal{G}$ with $7$ nodes, and the topology matrix $W$ is designed by the Laplacian rule \eqref{eq:topo-rule}. 
For simplicity without losing generality, the initial states of all nodes are also randomly generated, and the added noises $\{\theta_t\}_{t=0}^{T-1}$ are independent zero-mean Gaussian noises, whose variances are given by $\sigma_t^2=\frac{1}{(t+1)^\alpha}$. 
All the plots in each figure are drawn by repeating $20$ times of the iteration process \eqref{eq:global_model} or \eqref{eq:state_eta} from a common initial state and calculating their averages.

First, we examine the topology preservation performance by adding $\theta_t$ independently at each iteration. 
To explicitly present the preservation performance for the optimal decaying parameter $\alpha^*=\frac{1+\log 2}{\log 2}\approx 2.443$ in this case, we consider two groups of the variance decaying parameters, $\alpha\in\{0,0.5,1,1.5,2\}$ and $\alpha\in\{2.443,4,6,8,10\}$.
Fig.~\ref{fig:independent} shows the deviation between the actual and the idea states and the topology inference error, letting the variance decaying parameter range from $0$ to $10$. 
Concerning the state deviation $\|x_{t} - x^*_t \|^2$, it is easy to find from Fig.~\ref{fig:independ_state_0}-\ref{fig:independ_state_2} that $\|x_{t} - x^*_t \|^2$ will gradually grow with $t$ increasing when $\alpha\le1$, 
while remaining bounded when $\alpha>1$, corresponding to the conclusion \eqref{eq:state_divergence}. 
As the inference error $\|\hat{W}(T)-W\| $, Fig.~\ref{fig:independ_topo_0} compares the error curves when the variance decaying parameter $\alpha\in\{0,0.5,1,1.5,2\}$, and Fig.~\ref{fig:independ_topo_2} compares that when $\alpha\in\{2.443,4,6,8,10\}$
It is clear that $\|\hat{W}(T)-W\|$ possesses a slower decaying rate for a larger $\alpha$ in $[0,2.443]$, while a faster decaying rate for a larger $\alpha$ in $[2.443,\infty)$, which verifies the conclusions of Theorem \ref{th:decaying_rate}. 
It needs to be pointed out that in Fig.~\ref{fig:independ_topo_2}, the error curves generally remain stable and do not exhibit the decreasing trend to zero as Theorem \ref{th:decaying_rate} reveals. 
We observe that this consequence is incurred by the limitation of the calculating precision of the computer, and it will treat the added noises as zero when $T$ grows for a larger $\alpha$, thus making the curves almost unchanged. 
For example, when $\alpha=4$ and $T=49$, the variance of added noises is $\frac{1}{50^4}=1.6\times 10^{-7}$, which is extremely small. 

Next, Fig.~\ref{fig:alpha_curve} depicts the inference error of $\hat{W}$ regarding the decaying parameter $\alpha$, where 5 groups of curves are drawn to indicate the performance under different observation number $T$. 
It is easy to see from an arbitrary curve with a fixed $T$ that, the inference error will increase first and then decrease with $\alpha$ growing. 
Specifically, the error reaches the maximum value near $\alpha^*\approx 2.443$, which verifies the conclusions of Theorem \ref{th:optimal}. 
Here we would like to observe that since $R_{\theta}(T)$ is computed by the variance and expectation quantities, the optimality of $\alpha^*$ for $R_{\theta}(T)$ does not hold for $\| \hat{W}(T)\!-\!W\|$ almost surely but in high probability. 
Therefore, it is possible that the maximum value of $\| \hat{W}(T)\!-\!W\|$ in practice does not correspond to $\alpha^*$ exactly. 
Additionally, our extensive tests show that the optimal $\alpha$ for $\|\hat{W}(T)-W\|$ generally locates in the interval $[2,3]$, which is sufficient to reflect the effectiveness of $R_{\theta}(T)$.

\begin{figure}[t]
\centering
\includegraphics[width=0.4\textwidth]{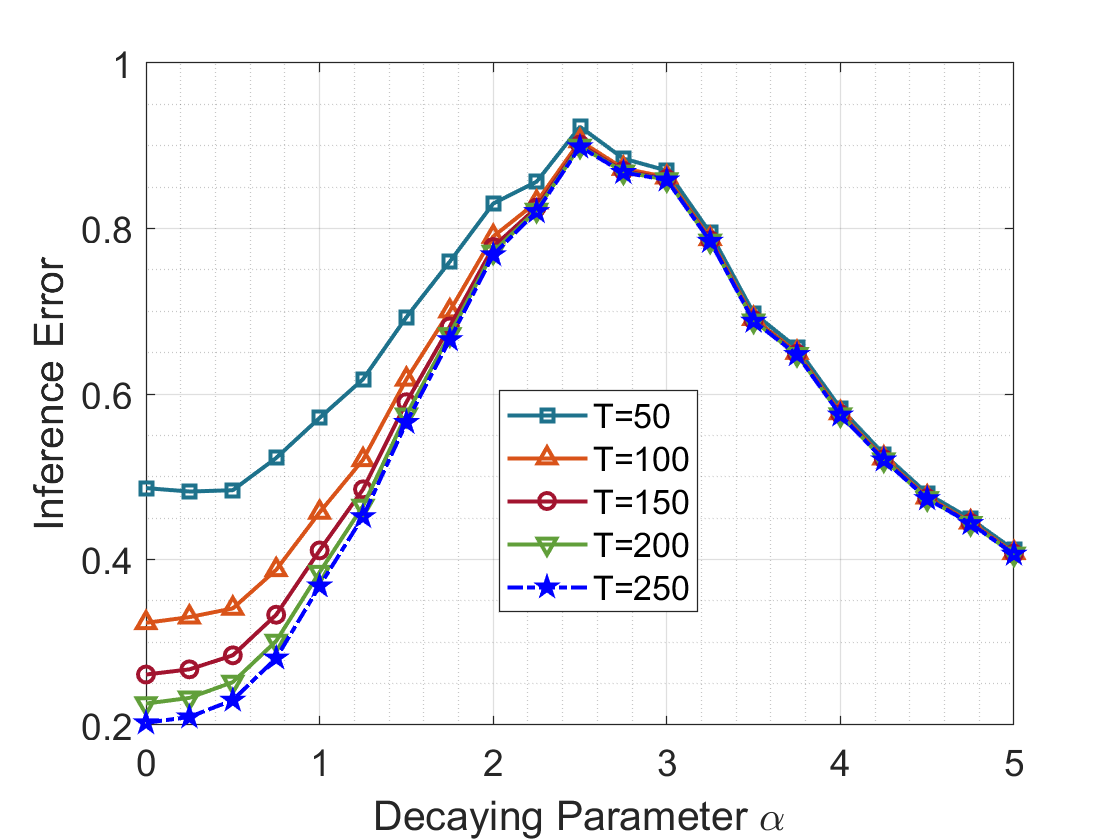}
\caption{The topology inference regarding the decay parameter $\alpha$, given fixed observation number $T$.}
\label{fig:alpha_curve}
\vspace{-15pt}
\end{figure}

Finally, we focus on the preservation performance evaluation by adding $\xi_t=\theta_t - \theta_{t-1}$, which has one-lag time dependence with each other. 
To explicitly exhibit the zero state-deviation in the asymptotic sense, in this case we consider a longer time horizon $T=10^5$ and use logarithmic coordinates in the time dimension. 
Similar to the case of adding $\theta_t$, we consider two groups of the variance decaying parameters, $\alpha\in\{0,0.5,1,1.5,2\}$ and $\alpha\in\{2.443,4,6,8,10\}$.
As shown in Fig.~\ref{fig:depend_state_0}-\ref{fig:depend_state_2}, except that the state deviation curve with non-decreasing variance parameter $\alpha=0$ is strictly bounded, all other curves decay to zero asymptotically. 
Specifically, the larger $\alpha$ is, the faster $\|x_{\xi,t} - x^*_t \|^2$ will decay to zero, which echoes with the conclusions of Theorem \ref{th:state_bound}. 
Concerning the inference error $\|\hat{W}(T)-W\| $, Fig.~\ref{fig:depend_state_0}-\ref{fig:depend_state_2} demonstrate that the inference errors will not decay to zero regardless of the setting of $\alpha$, which verifies the conclusion of Theorem \ref{th:rate_correlation}. 
Taking $\alpha=0$ as an example, the noise variance is non-decreasing and will not face the calculating precision issue of the computer as when $\alpha$ is large, and $\|\hat{W}(T)-W\| $ still will not go to zero as $T\to\infty$. 
It is worth noting that different from the case of adding $\theta_t$, the decaying rate of $\|\hat{W}(T)-W\| $ when adding $\xi_t$ does not possess explicit monotonicity regarding $\alpha$, as Theorem \ref{th:rate_correlation} reveals that it locates in an error bound. 
Hence, it is possible that the inference error when $\alpha=10$ is larger than that when $\alpha=2$.

\begin{figure*}[t]
\centering
\subfigure[$\alpha\in\{0,0.25,0.5,0.75,1\}$]{\label{fig:depend_state_0}
\includegraphics[width=0.26\textwidth]{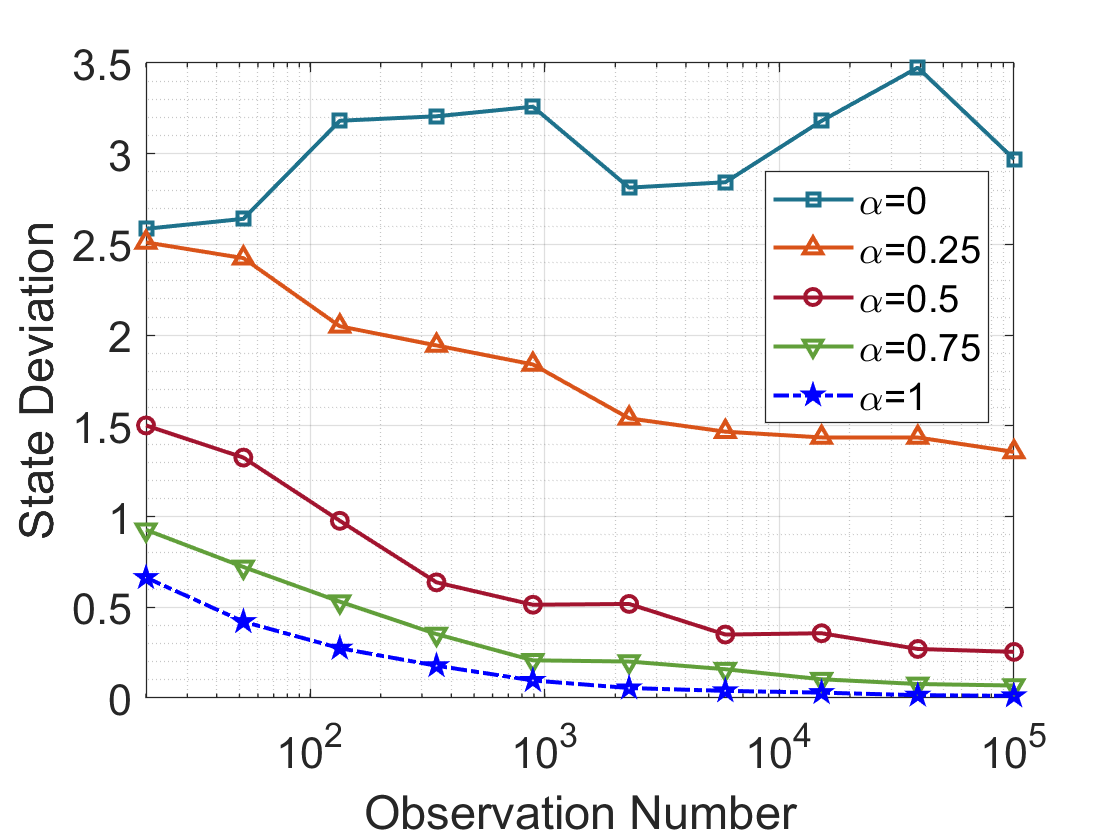}}
\hspace{-0.65cm}
\subfigure[$\alpha\in\{2,4,6,8,10\}$]{\label{fig:depend_state_2}
\includegraphics[width=0.26\textwidth]{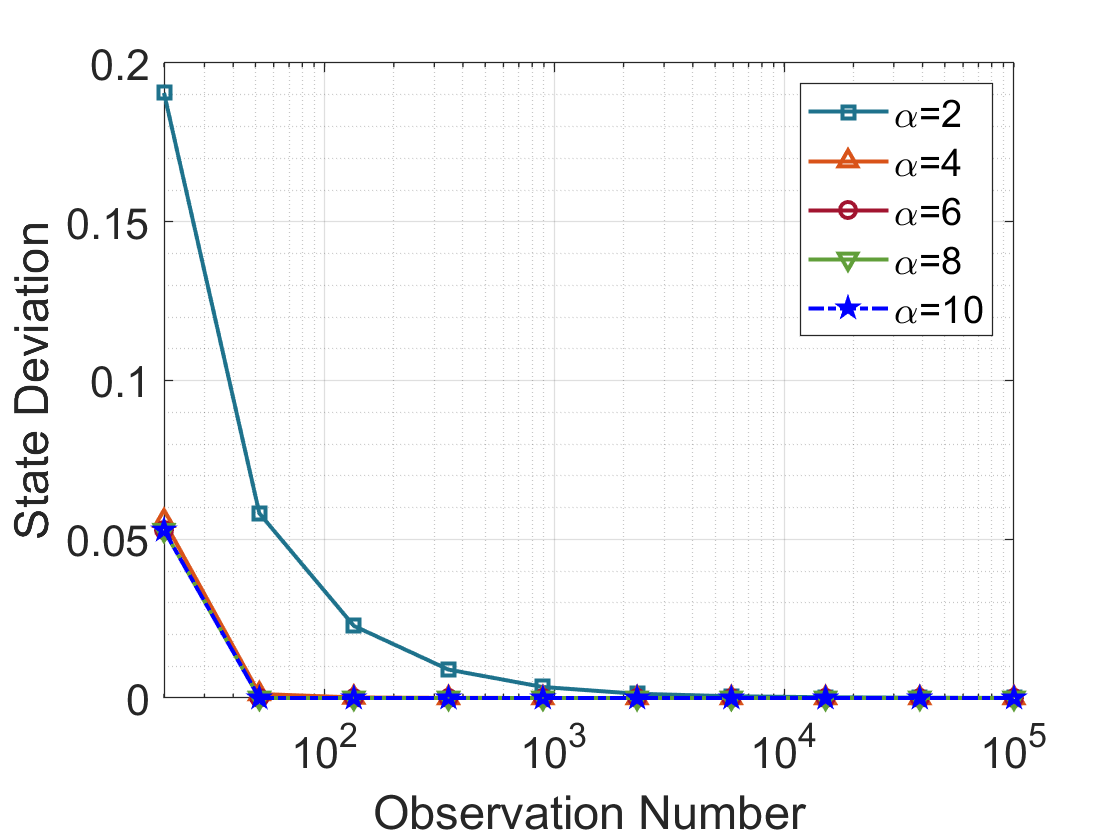}}
\hspace{-0.65cm}
\subfigure[$\alpha\in\{0,0.25,0.5,0.75,1\}$]{\label{fig:depend_topo_0}
\includegraphics[width=0.26\textwidth]{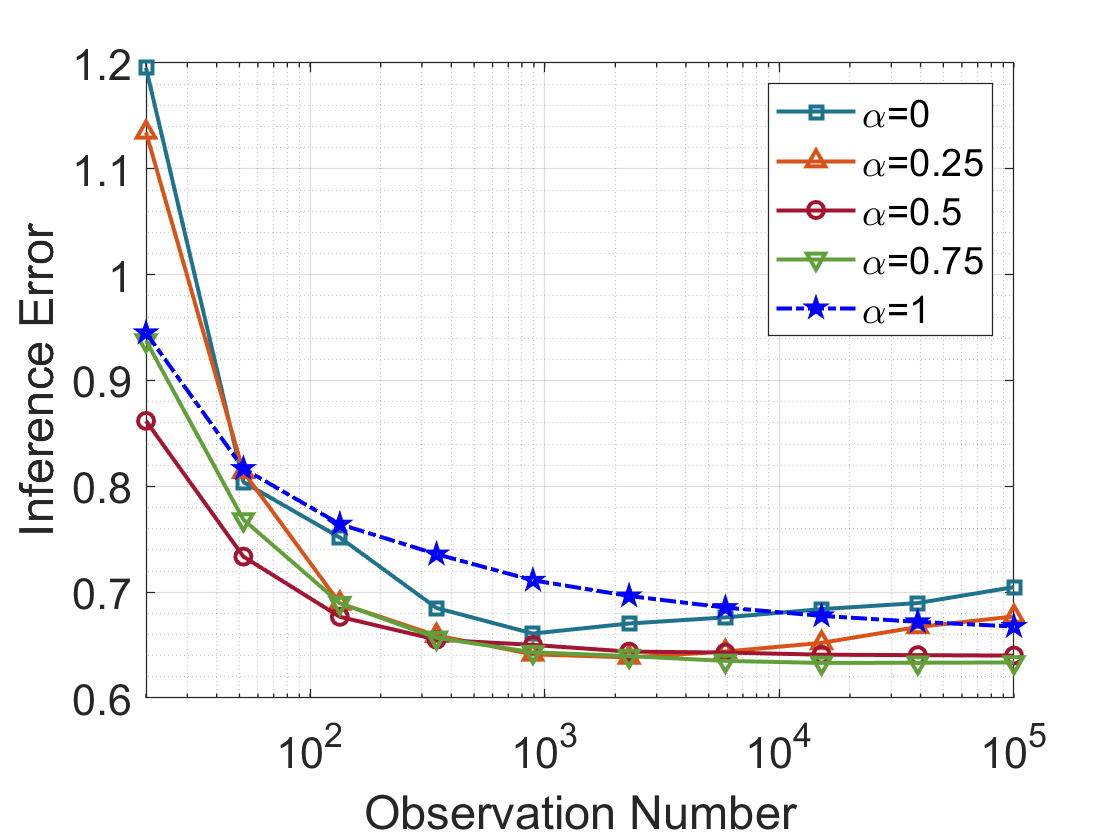}}
\hspace{-0.65cm}
\subfigure[$\alpha\in\{2,4,6,8,10\}$]{\label{fig:depend_topo_2}
\includegraphics[width=0.26\textwidth]{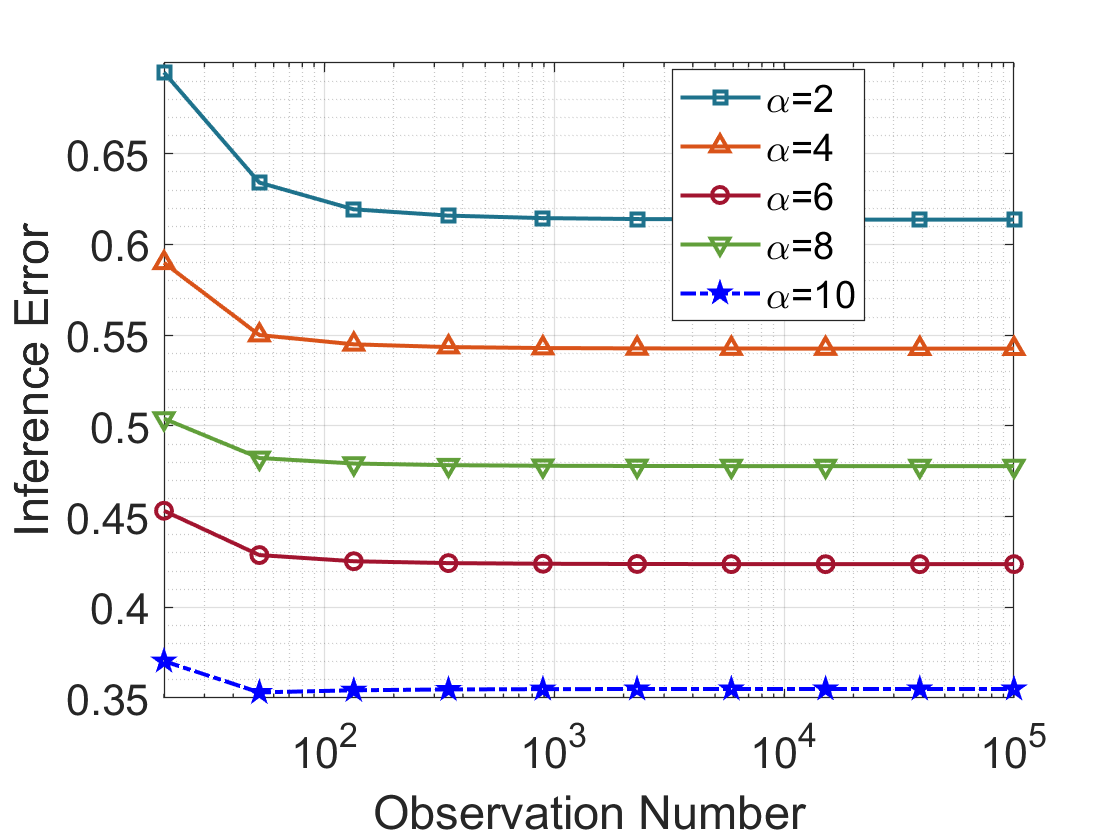}}
\vspace{-5pt}
\caption{The state deviation and topology inference error regarding $T$ when dependent noises $\{\xi_t\}$ are added. 
(a)(b): The state deviation $\|x_{\xi,t} - x^*_t \|^2$. (c)(d): The topology inference error $\|\hat{W}_{\xi}(T)-W\| $. }
\label{fig:dependent}
\vspace*{-13pt}
\end{figure*}

\section{Conclusions}\label{sec:conclusion}
In this paper, we studied the problem of preserving the topology of NSs by adding random noises. 
We first proposed a trace-based variance-expectation ratio motivated by the probability tails of some fundamental concentration inequalities. 
The metric did not depend on the distribution forms of the added noises and fully characterized the decaying rate of the topology inference error regarding the observation number. 
Then, based on the Euler-Maclaurin approximation, we analyzed the topology preservation performance of adding decaying independent noises and revealed the preservation limitation of zero asymptotic inference.  
The optimal noise design with the slowest decaying rate for the inference error was obtained. 
We further extended the metric to analyze the topology preservation performance under decaying noises with multi-lag time dependence, which achieved zero state deviation and non-zero inference error in the asymptotic sense simultaneously. 

The results of this paper provide a meaningful exploration of preserving the topology of NSs, and several promising directions include: 
i) extending the theoretical framework to NSs where each node is governed by a general linear-time-invariant dynamic model; 
ii) investigating the interplay between the added noises for topology preservation and the noises injected by the adversary for better inference.

\appendix
\subsection{Proof of Lemma \ref{le:exp_var}}\label{pr:le:exp_var}
\begin{proof}
First, notice that the quadratic form $z=\bm{\theta}_{T}^\mathsf{T} Q \bm{\theta}_{T}$ is equivalent to $z= \sum_{\ell_1= 1}^{nT} \sum_{\ell_2= 1}^{nT} Q_{\ell_1\ell_2} \bm{\theta}_{\ell_1} \bm{\theta}_{\ell_2}$. 
For $\ell_1\neq \ell_2$, $\bm{\theta}_{\ell_1} $ and $\bm{\theta}_{\ell_2}$ are independent of each other (i.e., $\mathbb{E}[z](\bm{\theta}_{\ell_1}\bm{\theta}_{\ell_2} )=0 $), and thus we have
\begin{align}
\mathbb{E}[z] = &\sum_{\ell_1= 1}^{nT} \sum_{\ell_2= 1}^{nT} \mathbb{E}[Q_{\ell_1\ell_2} \bm{\theta}_{\ell_1} \bm{\theta}_{\ell_2}] =  \sum_{\ell= 1}^{nT} Q_{\ell\ell} \mathbb{D}[\bm{\theta}_{\ell}].  
\end{align}
As for $\mathbb{D}[z]$, it needs to be pointed out first that for the two term $\bm{\theta}_{\ell_{1a}} \bm{\theta}_{\ell_{2a}} $ and $\bm{\theta}_{\ell_{1b}} \bm{\theta}_{\ell_{2b}} $ ($\mathbbm{1}_{\ell_{1a}=\ell_{1b}} \mathbbm{1}_{\ell_{2a}=\ell_{2b}}\neq 1 $), it follows from the product expectation formula \eqref{eq:expectation_product} that 
\begin{align}\label{formula:zero_correlation}
&\operatorname{Cov}[Q_{\ell_{1a}\ell_{2a}}\bm{\theta}_{\ell_{1a}} \bm{\theta}_{\ell_{2a}},Q_{\ell_{1b}\ell_{2b}}\bm{\theta}_{\ell_{1b}} \bm{\theta}_{\ell_{2b}}]  \nonumber \\
&= \mathbb{E}[ Q_{\ell_{1a}\ell_{2a}} Q_{\ell_{1b}\ell_{2b}} \bm{\theta}_{\ell_{1a}} \bm{\theta}_{\ell_{2a}} \bm{\theta}_{\ell_{1b}} \bm{\theta}_{\ell_{2b}} ] \nonumber \\
&~~~- \mathbb{E}[Q_{\ell_{1a}\ell_{2a}}\bm{\theta}_{\ell_{1a}} \bm{\theta}_{\ell_{2a}}] \mathbb{E}[Q_{\ell_{1b}\ell_{2b}}\bm{\theta}_{\ell_{1b}} \bm{\theta}_{\ell_{2b}}]=0. 
\end{align}
Then, based on the sum-variance formula \eqref{formula:sum_variance}, substituting into \eqref{formula:zero_correlation} into $\mathbb{D}[z]$ yields 
\begin{align}
\!\!\!\mathbb{D}[z] \!=\!&  \!\sum_{\ell_1= 1}^{nT} \sum_{\ell_2= 1}^{nT} Q_{\ell_1\ell_2}^2 \mathbb{D}[ \bm{\theta}_{\ell_1} \bm{\theta}_{\ell_2}] \nonumber \\
&  + \!\!\!\!\! \!\!\!\!\! \!\!\!\!\!\! \!\!\! \sum_{\substack{ \ell_{1a},\ell_{2a},\ell_{1b}, \ell_{2b}=1 \\ \mathbbm{1}_{\ell_{1a}=\ell_{1b}} \mathbbm{1}_{\ell_{2a}=\ell_{2b}}\neq 1 }}^{nT} \!\!\!\!\!\! \!\!\! \!\!\!\!\!\! \!\!\! \operatorname{Cov}[Q_{\ell_{1a}\ell_{2a}}\bm{\theta}_{\ell_{1a}} \bm{\theta}_{\ell_{2a}},Q_{\ell_{1b}\ell_{2b}}\bm{\theta}_{\ell_{1b}} \bm{\theta}_{\ell_{2b}}] \nonumber \\
=& \sum_{\ell= 1}^{nT} Q_{\ell\ell}^2 \mathbb{D}[\bm{\theta}_{\ell}^2]  + \sum_{\ell_1,\ell_2= 1, \ell_1\neq\ell_2}^{nT} Q_{\ell_1 \ell_2 }^2 \mathbb{D}[\bm{\theta}_{\ell_1} \bm{\theta}_{\ell_2} ] \nonumber \\
=& \sum_{\ell= 1}^{nT} Q_{\ell\ell}^2 (\mathbb{E}[\bm{\theta}_{\ell}^4] \!-\! \mathbb{D}^2[\bm{\theta}_{\ell}]) +  \!\! \!\!\!\!\!\! \!\!\!\sum_{\ell_1,\ell_2= 1, \ell_1\neq\ell_2}^{nT} \!\! \!\!\!\!\!\! \!\!\! Q_{\ell_1 \ell_2 }^2 \mathbb{D}[\bm{\theta}_{\ell_1}] \mathbb{D}[\bm{\theta}_{\ell_2}],
\end{align}
where the product variance formula \eqref{eq:variance_product} is applied in the last equality. 
The proof is completed. 
\end{proof}

\subsection{Proof of Theorem \ref{th:equivalent_rate}}\label{pr:th:equivalent_rate}
\begin{proof}
The key idea of this proof is to take an element-wise expectation and variance analysis on the matrix entries. 

\begin{itemize}
\item \textit{Part 1: Compute $\operatorname{tr}( \mathbb{D}\left[\Theta (X -\mathbb{E}[X])^\mathsf{T}  \right] )$ and $\operatorname{tr}( \mathbb{E}[(X \!-\! \mathbb{E}[X]) (X \!-\! \mathbb{E}[X])^\mathsf{T}] )$. }
\end{itemize}

First, based on the expansion of $x_t$ by \eqref{eq:x_expansion}, the matrix $\Theta (X -\mathbb{E}[X])^\mathsf{T}$ is written as 
\begin{align}\label{eq:expand_cx}
\Theta (X -\mathbb{E}[X])^\mathsf{T} &=  \sum_{t = 0}^{T-1} \theta_t (x_t - \tilde{x}_t )^\mathsf{T} =  \sum_{t = 1}^{T-1} \theta_t \tilde\theta_{t-1}^\mathsf{T}.  
\end{align}
To characterize the variance of $\theta_t \tilde\theta_{t-1}^\mathsf{T}$, we need to obtain the expectation and variance of $\tilde\theta_{t-1}$, given by 
\begin{align}
\!\mathbb{E}[\tilde{\theta}_{t-1}]\!=\!0,~
\mathbb{D}[\tilde{\theta}_{t-1}] \!=\! 
\sum_{m = 0}^{t-1} W^m (W^m)^\mathsf{T}\sigma_0^2 \! =\! \sigma_0^2 \Gamma_t.  
\end{align}
Due to the independence between $\theta_t$ and $\tilde{\theta}_{t-1}$, the product of two elements, $\theta_t^{i} \tilde{\theta}_{t-1}^{j}$, satisfies
\begin{align}
\!\! \!\mathbb{E}[\theta_t^{i} \tilde{\theta}_{t-1}^{j}] &\!=\!\mathbb{E}[\theta_t^{i}] \mathbb{E}[\tilde{\theta}_{t-1}^{j}] =0, \label{eq:zero_exp}\\
\!\! \! \mathbb{D}[\theta_t^{i} \tilde{\theta}_{t-1}^{j}] &\!=\! \mathbb{D}[\theta_t^{i}] \mathbb{D}[\tilde{\theta}_{t-1}^{j}] \!+\! \mathbb{D}[\theta_t^{i}] \mathbb{E}^2[\tilde{\theta}_{t-1}^{j}] \!+\!  \mathbb{E}^2[\theta_t^{i}]  \mathbb{D}[\tilde{\theta}_{t-1}^{j}] \nonumber \\
&\!=\! \sigma_0^4 \Gamma_t^{jj}. 
\end{align}
Based on \eqref{eq:zero_exp}, we show that the correlation between arbitrary two elements $\theta_{t_1}^{i} \tilde{\theta}_{t_1-1}^{j}$ and $\theta_{t_2}^{i} \tilde{\theta}_{t_2-1}^{j}$ ($t_1 \neq t_2$) is zero, i.e., 
\begin{align}\label{eq:zero_cor1}
\mathbb{E}\left[ (\theta_{t_1}^{i} \tilde{\theta}_{t_1-1}^{j} - \mathbb{E}[\theta_{t_1}^{i} \tilde{\theta}_{t_1-1}^{j}]  ) (\theta_{t_2}^{i} \tilde{\theta}_{t_2-1}^{j}- \mathbb{E}[\theta_{t_2}^{i} \tilde{\theta}_{t_2-1}^{j}] ) \right ]=0,
\end{align}
which is guaranteed by the independence between $\theta_{t_1}^{i}$ and $\theta_{t_2}^{j} $. 
Following the zero correlation properties \eqref{eq:zero_cor1}, the variance of $(\sum_{t = 1}^{T-1} \theta_t \tilde\theta_{t-1}^\mathsf{T})^{ij}$ is given by 
\begin{align}\label{eq:variance_zero}
\mathbb{D}[(\sum_{t = 1}^{T-1} \theta_t \tilde\theta_{t-1}^\mathsf{T})^{ij}] &= \sum\limits_{t = 1}^{T-1} \mathbb{D}[\theta_t^{i} \tilde{\theta}_{t-1}^{j}] = \sum\limits_{t = 1}^{T-1} \sigma_0^4 \Gamma_t^{jj}.  
\end{align}
Then, we have 
\begin{align}\label{eq:tr_dx0}
\operatorname{tr}( \mathbb{D}\left[\Theta (X -\mathbb{E}[X])^\mathsf{T}  \right] ) = \sigma_0^4 \sum\limits_{t = 1}^{T-1} \operatorname{tr}( \Gamma_t).
\end{align}

Next, taking the expansion of $x_t$ by \eqref{eq:x_expansion} into $(X \!-\! \mathbb{E}[X]) (X \!-\! \mathbb{E}[X])^\mathsf{T}$, we have
\begin{align}
&(X \!-\! \mathbb{E}[X]) (X \!-\! \mathbb{E}[X])^\mathsf{T} \nonumber \\
= &\!\sum_{t = 1}^{T-1} (\tilde{x}_t \!+\! \tilde{\theta}_{t-1} \!-\! \tilde{x}_t ) (\tilde{x}_t \!+\! \tilde{\theta}_{t-1} \!-\! \tilde{x}_t )^\mathsf{T} \!=\! \sum_{t = 1}^{T-1} \tilde{\theta}_{t-1}\tilde{\theta}_{t-1}^\mathsf{T}, \label{eq:expectation_minus}\\
\Rightarrow & \operatorname{tr}( \mathbb{E}[(X \!-\! \mathbb{E}[X]) (X \!-\! \mathbb{E}[X])^\mathsf{T}] ) = \sigma_0^2 \sum_{t = 1}^{T-1} \operatorname{tr}( \Gamma_t). \label{eq:tr_ex0}
\end{align}
Then, divide the square-root of \eqref{eq:tr_dx0} with \eqref{eq:tr_ex0}, and we obtain 
\begin{align}
R_{\theta}(T) \!=\! \frac{\sqrt{\operatorname{tr}( \mathbb{D}\left[\Theta (X -\mathbb{E}[X])^\mathsf{T}  \right] )}}{\operatorname{tr}( \mathbb{E}[(X \!-\! \mathbb{E}[X]) (X \!-\! \mathbb{E}[X])^\mathsf{T}] )} \!=\! \frac{1}{\sqrt{ \operatorname{tr}(\sum_{t = 1}^{T-1} \Gamma_t)}}.
\end{align}
It is worth noting that based on \eqref{eq:expand_cx} and \eqref{eq:expectation_minus}, one can easily obtain that $R_{\theta}(T)$ is equivalent to $\frac{ \sqrt{ \operatorname{tr}(\mathbb{D}\left[\Theta X^\mathsf{T}  \right] )}  }{ \operatorname{tr}(\mathbb{E}[X X^\mathsf{T}] ) }$ with $x_0=0$. 

\begin{itemize}
\item \textit{Part 2: Characterize $\bm{O}(R_{\theta}(T))$. }
\end{itemize}

For simple notation, suppose that $W$ has $n$ distinct eigenvalues, and the eigenvalue decomposition of $W$ is given by 
\begin{align}\label{eq:eig_decomposition}
W=MJM^{-1}=M \operatorname{diag}\{\lambda_1,\lambda_2,\cdots,\lambda_n\} M^{-1},
\end{align}
where $M\in\mathbb{R}^{n \times n}$ is an invertible matrix. 
Substituting the eigenvalue decomposition \eqref{eq:eig_decomposition} into $W^m$ yields that $W^m= M J^m M^{-1}$, satisfying 
\begin{equation}\label{eq:upper_bound_W}
\!\!\|W^m \|_F \!\le\! \sqrt{n} \|W^m\| \!\le\! \sqrt{n} \kappa(M) \|J^m\| \!\le\! n^{\frac{3}{2}} \kappa(M) |\lambda_1|^m,
\end{equation}
where $\kappa(M)=\|M\| \|M^{-1}\|$ is the condition number of $M$ and the property $\|J^m\|\le n \|J^m\|_{\max}=n |\lambda_1|^m $ is used. 
Then, based on \eqref{eq:upper_bound_W} and $\rho(W^m) \le \|W^m\|\le \|W^m\|_F$, it follows that $\operatorname{tr}(W^m (W^m)^\mathsf{T})=\|W^m\|_F^2$ is bounded by 
\begin{align}\label{eq:bounds_w}
\!\! |\lambda_1|^{2m} \!=\!\rho^2(W^m) \!\le\! \operatorname{tr}(W^m (W^m)^\mathsf{T}) \!\le\! n^3 \kappa^2(M) |\lambda_1|^{2m}.
\end{align}
It is clear from \eqref{eq:bounds_w} that $\operatorname{tr}(W^m (W^m)^\mathsf{T})$ is dominated by $|\lambda_1|^{2m}$. 
Note that the property \eqref{eq:bounds_w} still holds even when the eigenvalues are not distinct, because the their corresponding Jordan blocks can always be bounded by $\|J^m\|\le n \|J^m\|_{\max}=n |\lambda_1|^m $. 

Therefore, based the above analysis, the increment scale of $\operatorname{tr}(\Gamma_t) $ is the same as that of $\sum\nolimits_{m = 0}^{t-1} \rho^{2m}(W)$, and we have 
\begin{align}
\bm{O}(\operatorname{tr}(\sum\limits_{t = 1}^{T-1} \Gamma_t)) &= \bm{O}(\sum\limits_{t = 1}^{T-1} \operatorname{tr}(\Gamma_t)) =\bm{O}(\sum\limits_{t = 1}^{T-1} \sum\limits_{m = 0}^{t-1} \rho^{2m}(W)) \nonumber \\
&=\left\{\begin{aligned}
&T,&&\text{if}~ \rho(W)<1\\
&T^2, &&\text{if}~ \rho(W)=1\\
&\rho^{2T}(W),&&\text{if}~ \rho(W)>1
\end{aligned}~.\right.
\end{align}
Since $R_{\theta}(T)\!=\! \frac{1}{\sqrt{ \operatorname{tr}(\sum_{t = 1}^{T-1} \Gamma_t)}}$, the increment scale of $R_{\theta}(T)$ is given by 
\begin{align}
R_{\theta}(T)={1}/\sqrt{\bm{O}\left(\operatorname{tr}(\sum\nolimits_{t = 1}^{T-1} \Gamma_t)\right)},
\end{align}
which equals to \eqref{eq:rates_eq}. 
The proof is completed. 
\end{proof}

\subsection{Proof of Lemma \ref{le:approximate}}\label{pr:le:approximate}
\begin{proof}
To begin with, we consider the simplest situation where $\alpha>1$. 
It is well-known that $h_{\alpha}(T)$ with $\alpha>1$ is a typical p-series that is strictly bounded, i.e., 
\begin{equation}
h_{\alpha}(T)=\bm{O}(C),
\end{equation}
where $C>0$ is a constant irrelevant to $T$. 
Apparently, the increment scale of $h_{\alpha}(T)$ is identical with $\int_1^T \frac{\sigma_0^2}{y^\alpha} dy =\bm{O}( C)$. 

When $\alpha\le1$, $h_{\alpha}(T)$ will diverge as $T$ grows and obtaining its accurate expression is intractable. 
Since $f(t)$ is an infinity-order continuous function on the interval $[1,T]$, we can apply the famous Euler-Maclaurin summation formula \cite{apostol1999elementary} on $ h_{\alpha}(T)$ and obtain 
\begin{align}\label{eq:sum_Mac}
h_{\alpha}(T)\! =\!F_{\alpha}(T)\!+\!\sum_{r=0}^p \frac{(-1)^{r\!+\!1} B_{r\!+\!1}}{(r\!+\!1)!}(f^{(r)}(T)\!-\!f^{(r)}(1))\!+\! R_p,
\end{align}
where $B_r$ is the $r$-th Bernoulli number, $R_p$ is an error term (which depends on $T$ and $p$, and $f$) and is usually small for suitable values of $p$. 
Then, we can characterize the increment characteristic of $h_{\alpha}(T)$ by computing the above asymptotic expansions of $h_{\alpha}(T)$. 
By moving the term $F_{\alpha}(T)$, we have 
\begin{equation}\label{eq:ex_difference}
h_{\alpha}(T)\! - \!F_{\alpha}(T)\!=\!\sum_{r=0}^{\infty} \frac{(-1)^{r\!+\!1} B_{r\!+\!1}}{(r\!+\!1)!}(f^{(r)}(T)\!-\!f^{(r)}(1)).
\end{equation}
In the sequel, we will show that the series in the RHS of \eqref{eq:ex_difference} is bounded. 
Note that the Bernoulli number $B_r$ satisfies the following series identity
\begin{equation}\label{eq:converge_b}
\sum_{r=1}^{\infty} \frac{ (-1)^{r+1} B_{r+1} }{(r+1)!} \!=\! \sum_{r=0}^{\infty} \frac{ (-1)^{r+1} B_{r+1} }{(r+1)!} \!-\! \frac{1}{2}= \frac{1}{1\!-\!\frac{1}{e}} \!-\!\frac{1}{2}. 
\end{equation}
For the $r$-th derivative, we have $f^{(r)}(T)= \frac{(-1)^{r}(\alpha+r)!}{T^{\alpha+r}}$, which is bounded when $T$ is sufficiently large. 
This point can be verified by considering the infinity-order derivative when $T\to \infty$, i.e.,
\begin{equation}\label{eq:converge_f}
\mathop {\lim }\limits_{T \to \infty} f^{(T)}(T)= \mathop {\lim }\limits_{T \to \infty} \frac{(-1)^{T}(\alpha+T)!}{T^{\alpha+T}} =0,
\end{equation}
where the Stirling's approximation $\frac{T !}{T^T} \approx \sqrt{2 \pi T}\left(\frac{1}{e}\right)^T$ is used. 
Then, it follows from the convergence of \eqref{eq:converge_b} and the near-exponential convergence of \eqref{eq:converge_f} that 
\begin{align}\label{eq:bounded_approximation}
\mathop {\lim }\limits_{T \to \infty}  \sum_{r=0}^{\infty} \frac{(-1)^{r\!+\!1} B_{r\!+\!1}}{(r\!+\!1)!}(f^{(r)}(T)\!-\!f^{(r)}(1))<\infty.
\end{align}
Finally, based on \eqref{eq:ex_difference} and \eqref{eq:bounded_approximation}, we have $h_{\alpha}(T)=\bm{O}(F_{\alpha}(T))$ and complete the proof. 
\end{proof}

\subsection{Proof of Theorem \ref{th:decaying_rate}}\label{pr:th:decaying_rate}
\begin{proof}
The proof is straightforward by characterizing $ \operatorname{tr}(\mathbb{D}\left[\Theta X^\mathsf{T}  \right] )$ and $\operatorname{tr}(\mathbb{E}[X X^\mathsf{T}] )$, respectively. 
The following notations are introduced to ease the analysis. 
\begin{align}
& \bm{x}_{0:T-1}\!=\![x_0^\mathsf{T},x_1^\mathsf{T},\cdots,x_{T-1}^\mathsf{T}]^\mathsf{T}, \Gamma_{t_1,t}^*(\ell) \! = \!\!\!\!\! \sum_{m = t_1-1}^{t-1} \!\!\!  (W^m)^\mathsf{T} W^{m+\ell},   \nonumber  \\
&\tilde{W}\!=\!\begin{bmatrix}
0 & 0  & \ldots & 0 \\
I & 0  & \ldots & 0 \\
\vdots & \vdots & \ddots & \vdots \\
W^{T-2} & W^{T-3} & \ldots & 0
\end{bmatrix}, \label{eq:tilde_w} \\
&Q_w \!=\!\tilde{W}^\mathsf{T} \tilde{W} \nonumber \\
&~\!=\!\!\begin{bmatrix}
\Gamma_{T\!-\!1}^*    & \Gamma_{1,T\!-\!2}^{*\mathsf{T}}(1)  &  \ldots & \Gamma_{1,1}^{*\mathsf{T}}(T\!-\!2) & 0 \\
\Gamma_{1,T\!-\!2}^{*}(1) & \Gamma_{T-2}^* &   \ldots & \Gamma_{1,1}^{*\mathsf{T}}(T\!-\!3)  & 0 \\
\vdots & \vdots & \ddots  & \vdots  & \vdots\\
\Gamma_{1,1}^{*}(T\!-\!2) & \Gamma_{1,1}^{*}(T\!-\!3)  &   \ldots & \Gamma_{1}^* & 0 \\
0 & 0  &   \ldots & 0 & 0 \\
\end{bmatrix}.  \label{eq:Q_W}
\end{align}

\begin{itemize}
\item \textit{Part 1: Characterize $ \operatorname{tr}(\mathbb{D}\left[\Theta X^\mathsf{T}  \right] )$. }
\end{itemize}

Note that although the random noises $\{\theta_t\}_{t=0}^{T-1}$ are no longer identically distributed under noise setting \eqref{eq:noise_variance}, their mutual independence is still ensured and will not incur any correlation terms when calculating the variance of $(\Theta  X^\mathsf{T})^{ij}$, i.e., 
$\mathbb{D}[(\Theta  X^\mathsf{T})^{ij}] = \sum\nolimits_{t = 1}^{T-1}\mathbb{D}[\theta_t^{i} \tilde{\theta}_{t-1}^{j}] $
 still holds. 
First, based on the condition that $\rho(W)=1$ and $\sigma_t^2=\frac{\sigma_0^2}{{(t+1)}^\alpha}$, we have 
\begin{align}
& \mathbb{D}[\tilde{\theta}_{t-1}] \!=\! \mathbb{E}[\tilde{\theta}_{t-1} \tilde{\theta}_{t-1}^\mathsf{T}]  \!=\!\! \sigma_0^2  \sum_{m = 0}^{t-1} \frac{ W^{t-m-1} (W^{t-m-1})^\mathsf{T} }{(m+1)^{\alpha}}, \\
&  \mathbb{D}[\theta_t^{i} \tilde{\theta}_{t-1}^{j}] = \mathbb{D}[\theta_t^{i}]  \mathbb{D}[\tilde{\theta}_{t-1}^{j}] = \frac{\sigma_0^4  \tilde{\Gamma}_{t,\alpha}^{jj} }{{(t+1)}^\alpha},
\end{align} 
where $\tilde{\Gamma}_{t,\alpha}\!=\!\sum\limits_{m = 0}^{t-1} \frac{ W^{t-m-1} (W^{t-m-1})^\mathsf{T} }{(m+1)^{\alpha}}$. 
Similar to \eqref{eq:zero_cor1}, due to the independence between $\theta_{t_1}^{i}$ and $\theta_{t_2}^{i} $ ($t_1 \neq t_2$), 
the correlation between arbitrary two elements $\theta_{t_1}^{i} \tilde{\theta}_{t_1-1}^{j}$ and $\theta_{t_2}^{i} \tilde{\theta}_{t_2 -1 }^{j}$ is zero. 
Hence, it follows that 
\begin{align}\label{eq:new_variance}
\mathbb{D}[(\Theta  X^\mathsf{T})^{ij}] =\sum\limits_{t = 1}^{T-1}  \mathbb{D}[\theta_t^{i} \tilde{\theta}_{t-1}^{j}] = \sum\limits_{t = 1}^{T-1} \frac{\sigma_0^4  \tilde{\Gamma}_{t,\alpha}^{jj} }{{(t+1)}^\alpha},
\end{align}
where $\tilde{\Gamma}_{t,\alpha}^{jj}$ dominates the increment of $\mathbb{D}[(\Theta  X^\mathsf{T})^{ij}]$. 
Recall that in the proof of Theorem \ref{th:equivalent_rate}, we have demonstrated that $|\lambda_1|^{2m}\le\operatorname{tr}(W^m (W^m)^\mathsf{T}) \!\le\! n^2 |\lambda_1|^{2m}$. 
When $\rho(W)=|\lambda_1|=1$, $\operatorname{tr}(\tilde{\Gamma}_{t,\alpha})$ is also bounded by
\begin{align}\label{eq:tilde_gamma}
\sum\limits_{m = 0}^{t-1} \!\frac{ 1 }{(m\!+\!1)^{\alpha}} \!\le\! \operatorname{tr}(\tilde{\Gamma}_{t,\alpha})  \!\le\! \sum\limits_{m = 0}^{t-1} \! \frac{ n^2 }{(m\!+\!1)^{\alpha}}. 
\end{align}
Substituting \eqref{eq:tilde_gamma} into $\operatorname{tr}(\mathbb{D}[(\Theta  X^\mathsf{T})]) = \sum\nolimits_{t = 1}^{T-1} \frac{\sigma_0^4  \operatorname{tr}(\tilde{\Gamma}_{t,\alpha}) }{{(t+1)}^\alpha}$, it follows that 
\begin{align}\label{eq:new_trace_dx}
\sigma_0^4  \sum\limits_{t = 1}^{T-1} \frac{ \sum\limits_{m = 0}^{t-1} \frac{1}{(m\!+\!1)^{\alpha} }} {{(t\!+\!1)}^\alpha} \!\le\!
\operatorname{tr}(\mathbb{D}[\Theta  X^\mathsf{T}]) \!\le\! n^2 \sigma_0^4  \sum\limits_{t = 1}^{T-1} \frac{ \sum\limits_{m = 0}^{t-1} \frac{1}{(m\!+\!1)^{\alpha} }} {{(t\!+\!1)}^\alpha}. 
\end{align}
It is clear from \eqref{eq:new_trace_dx} that the increment scale of $\operatorname{tr}(\mathbb{D}[\Theta  X^\mathsf{T}])$ regarding $T$ is the same as $ \sum\limits_{t = 1}^{T-1} \frac{ \sum\nolimits_{m = 0}^{t-1} {1}/{(m\!+\!1)^{\alpha} }} {(t\!+\!1)^\alpha}$, which contains double summation. 
Note that for the inner summation $\sum\nolimits_{m = 0}^{t-1} {1}/{(m\!+\!1)^{\alpha} }$, it cannot be characterized by $F_{\alpha}(t)$ alone because $\sum\nolimits_{m = 0}^{t-1} {1}/{(m\!+\!1)^{\alpha} }>0$ holds for arbitrary $t\ge1$. 
Typically, $F_{\alpha}(1)=0 $ will contradict with $\sum\nolimits_{m = 0}^{0} {1}/{(m\!+\!1)^{\alpha} }=1$, which indicates that the constant part cannot be omitted when doing the outer summation. 
Therefore, based on the nonnegativity, it holds that 
\begin{align}\label{eq:o_dx}
\bm{O}( \operatorname{tr}(\mathbb{D}\left[\Theta X^\mathsf{T}  \right] )= \bm{O}( \sum\limits_{t = 1}^{T-1} \frac{\operatorname{tr}(\tilde{\Gamma}_{t,\alpha})}{(t\!+\!1)^\alpha} )\!=\!\bm{O}(S_T),
\end{align}
where $S_T=\sum\nolimits_{t = 1}^{T-1} \frac{ \int_1^t 1/y^\alpha dy +C } {(t+1)^\alpha}$. 

Finally, the proof turns to characterize $\bm{O}(S_T)$. 
Note that $t<t+1$ and $\frac{1}{2t}<\frac{1}{t+1}$, and thus $S_T$ is bounded by
\begin{align}\label{eq:sum_s}
\!\!\left\{\begin{aligned}
\sum\limits_{t = 1}^{T-1} \frac{ t^{1\!-\!2\alpha} \!+\! C t^{-\alpha} }{2^{\alpha}(1-\alpha)} \!< &S_T \!<\! \sum\limits_{t = 2}^{T} \frac{ t^{1\!-\!2\alpha} \!+\!C t^{-\!\alpha} }{1-\alpha}, &&\alpha \!<\! 1 \\
\sum\limits_{t = 1}^{T-1} \frac{ \log t +C }{2t} \!< &S_T < \sum\limits_{t = 2}^{T} \frac{ \log t +C }{t},~&&\alpha=1 \\
\sum\limits_{t = 1}^{T-1} \frac{ C }{ 2^{\alpha} t^\alpha } \!< &S_T \!<\! \sum\limits_{t = 2}^{T} \frac{ C }{ t^\alpha },~&&\alpha > 1 
\end{aligned}
\right. \!\!.
\end{align}
From \eqref{eq:sum_s}, when $\alpha\le1$, $S_T$ will increase to infinity as $T$ grows, and thus we can omit the constant terms in $S_T$ and characterize $\bm{O}(S_T)$ by
\begin{small}
\begin{align}\label{eq:fenzi1}
\bm{O}( S_T ) \!=\! \left\{\begin{aligned}
&\bm{O}( \int_1^T \!\! {t^{1\!-\!2\alpha} \!+\! C t^{-\alpha}} dt ) \!=\! \bm{O}(  T^{2\!-\!2\alpha}) + \bm{O}( T^{1\!-\!\alpha} ), \alpha\!<\!1\\
&\bm{O}( \int_1^T \! {\frac{\log t \!+\! C}{t} } dt ) \!=\! \bm{O}( (\log T)^{2} )+ \bm{O}( \log T ) , \alpha\!=\!1
\end{aligned}\right.\!.
\end{align}
\end{small}
\!\!When $\alpha\!>\!1$, $S_T$ will converge to a constant determined by $\alpha$ as $T\!\to\!\infty$, and the terms concerning $T$ in $S_T$ can be omitted. 
Since the bounds of $S_T$ with $\alpha\!>\!1$ are characterized by 
\begin{align}
\left\{\begin{aligned}
&\bm{O}(\sum\limits_{t = 1}^{T-1} \frac{ C }{ 2^{\alpha} t^\alpha } ) \!=\! \bm{O}(  \frac{C}{2^\alpha} \int_1^{T\!-\!1} \!\! \frac{1}{y^\alpha} d y  ) \!=\! \bm{O}( \frac{C}{2^\alpha (\alpha\!-\!1)} ) \\
&\bm{O}(\sum\limits_{t = 2}^{T} \frac{ C }{ t^\alpha } ) \!=\! \bm{O}(  C \int_2^{T} \!\! \frac{1}{y^\alpha} d y  ) = \bm{O}( \frac{C}{2^{\alpha\!-\!1}(\alpha\!-\!1)} )
\end{aligned}\right.\!,
\end{align}
it follows that  
\begin{align}\label{eq:fenzi2}
\bm{O}(S_T)= \bm{O}(\frac{1}{2^\alpha (\alpha\!-\!1)} ),~\alpha>1.  
\end{align}
Note that different from \eqref{eq:fenzi1}, when $\alpha>1$, the coefficient concerning $\alpha$ characterized in \eqref{eq:fenzi2} cannot be omitted because $S_T$ converges to a constant that is directed determined by $\alpha$.


\begin{itemize}
\item \textit{Part 2: Characterize $\operatorname{tr}(\mathbb{E}[X X^\mathsf{T}] )$. }
\end{itemize}

Notice that $\operatorname{tr}(\mathbb{E}[X X^\mathsf{T}])=\mathbb{E}[\operatorname{tr}(X X^\mathsf{T})] $,  and we first show that $\operatorname{tr}(X X^\mathsf{T})$ is given by 
\begin{align}\label{eq:expectation_00}
\operatorname{tr}(X X^\mathsf{T}) = & \operatorname{tr}( \sum\nolimits_{t = 0}^{T-1} x_t x_t^\mathsf{T} ) = \sum\nolimits_{t = 0}^{T-1} x_t^\mathsf{T} x_t \nonumber \\  
=& \bm{x}_{0:T-1}^\mathsf{T}  \bm{x}_{0:T-1} = \bm{\theta}_{0:T\!-\!1}^\mathsf{T} Q_w  \bm{\theta}_{0:T\!-\!1} . 
\end{align}
Then, based on Lemma \ref{le:exp_var}, the expectation of \eqref{eq:expectation_00} is given by 
\begin{align}
\mathbb{E}[\operatorname{tr}(X X^\mathsf{T})] &= \mathbb{E}[\bm{\theta}_{0:T\!-\!1}^\mathsf{T} Q_w  \bm{\theta}_{0:T\!-\!1}]= \sum_{t = 1}^{T} \operatorname{tr}( Q_w(t,t)) \sigma_{t-1}^2   \nonumber \\
&= \sum_{t = 1}^{T-1} \operatorname{tr}( \Gamma_{T-t}^* ) \sigma_{t-1}^2. 
\end{align}
Recalling the inequality \eqref{eq:bounds_w}, we have $(T\!-\!t) \!\le\! \operatorname{tr}(\Gamma_{T-t}^* )\!\le\! n^2 \kappa^2(M)(T\!-\!t) $. 
Then, it follows that $\bm{O}(\operatorname{tr}(\Gamma_{T-t}^* ))=\bm{O}(T-t)$ and 
\begin{align}\label{eq:fenmu}
\!\!\bm{O}(\mathbb{E}[\operatorname{tr}(X X^\mathsf{T})]) \!= & \bm{O}(\sum_{t = 1}^{T-1} (T\!-\!t) \sigma_{t\!-\!1}^2) \!=\!\bm{O}( T\!\! \sum_{t = 1}^{T-1} \frac{1}{t^\alpha} \!-\! \sum_{t = 1}^{T-1} {t^{1 \!-\!\alpha}} )  \nonumber \\
=& \left\{\begin{aligned}
& \bm{O}(T^{2-\alpha} ) ,&&\text{if}~\alpha<1 \\
& \bm{O}(T \log T) ,&&\text{if}~\alpha=1 \\
& \bm{O}( \frac{T}{\alpha-1}) ,&&\text{if}~\alpha >1
\end{aligned}
\right.,
\end{align}
where the increment scale of $\sum_{t = 1}^{T-1} {t^{1 \!-\!\alpha}}$ is always smaller than that of $T\sum_{t = 1}^{T-1} \frac{1}{t^\alpha}$, and thus is not the dominant factor in $\mathbb{E}[\operatorname{tr}(X X^\mathsf{T})]$. 

Finally, based on \eqref{eq:fenzi1}, \eqref{eq:fenzi2} and \eqref{eq:fenmu}, the rate characterization of $R_{\theta}(T)={ \sqrt{ \operatorname{tr}(\mathbb{D}\left[\Theta X^\mathsf{T}  \right] )}  }/{ \operatorname{tr}(\mathbb{E}[X X^\mathsf{T}] ) }$ is given by
\begin{align}\label{eq:final_ratio2}
\!\! R_{\theta}(T) & = \frac{\bm{O} \left(  \sqrt{ \operatorname{tr}(\mathbb{D}\left[\Theta X^\mathsf{T}  \right] ) } \right) }{\bm{O}( \operatorname{tr}(\mathbb{E}[X X^\mathsf{T}]  )} = \sqrt{  \frac{\bm{O}(S_T )} {\bm{O}( \operatorname{tr}(\mathbb{E}[X X^\mathsf{T}] )^2 )} } ,
\end{align}
which further yields the decaying rate \eqref{eq:limit_Re} for different $\alpha$ and completes the proof. 
\end{proof}

\subsection{Proof of Corollary \ref{coro:sufficient_noise}}\label{pr:coro:sufficient_noise}
\begin{proof}
The proof of this result is similar to that of Theorem \ref{th:decaying_rate}, and here we sketch the main idea. 
First, recalling the asymptotic expansion \eqref{eq:ex_difference} by the Euler-Maclaurin summation formula, to ensure the series $\sum\nolimits_{t = 0}^{T-1} g(t)$ can be approximated by $\int_0^T g(y) dy$, it requires the following series to converge, i.e., 
\begin{equation}\label{eq:ex_difference2}
\sum_{r=0}^{\infty} \frac{(-1)^{r\!+\!1} B_{r\!+\!1}}{(r\!+\!1)!}(g^{(r)}(T)\!-\!g^{(r)}(1))<\infty.
\end{equation}
Notice that the odd Bernoulli numbers $B_{2r+1}$ are zero except for $B_1$, and the even ones have the following asymptotic approximations
\begin{equation}
(-1)^{r+1} B_{2 r} \sim \frac{2(2 r) !}{(2 \pi)^{2 r}} \Rightarrow \frac{(-1)^{r\!+\!1} B_{2r}}{(2r)!} \sim \frac{2}{(2 \pi)^{2 r}}
\end{equation}
Hence, if the series \eqref{eq:ex_difference2} converges, it will converge in the way similar to a exponential series. 
Then, we can turn to ensure the last term of the series satisfy the following limit
\begin{align}\label{eq:key_limit}
&\mathop {\lim }\limits_{T \to \infty} \frac{ B_{2T} (g^{(2T)}(2T)\!-\!g^{(2T)}(1)) }{(2T)!} \nonumber \\
= & \mathop {\lim }\limits_{T \to \infty}  \frac{ B_{2T}(g^{(2T)}(2T)\!-\!g^{(2T)}(1)) }{(2 \pi)^{2 T}}=0,
\end{align}
which is exactly the derivative condition in \eqref{eq:converge_condition}. 

With \eqref{eq:key_limit} guaranteed, $\sum\nolimits_{t = 0}^{T-1} g(t)$ can be approximated $\int_0^{T-1} g(t)$. 
Resembling the deduction of \eqref{eq:o_dx} and \eqref{eq:fenmu}, the increment scale of $\|\mathbb{D}[(\Theta  X^\mathsf{T})]\|_F$  and  $\|\mathbb{E}[X X^\mathsf{T}]\|_F $ are characterized by 
\begin{align}
\operatorname{tr}(\mathbb{D} [\Theta X^\mathsf{T}]) &\!=\!\bm{O}(\int_1^{T\!-\!1} \!\! g(y_2) \int_0^{y_2\!-\!1} g(y_1)  d{y_1} d{y_2}), \label{eq:rate_g1}\\
\operatorname{tr}(\mathbb{E}[X X^\mathsf{T}] )&\!=\!\bm{O}( T \int_0^{T-2} g(y) d{y}) \label{eq:rate_g2}.
\end{align}
Finally, substituting \eqref{eq:rate_g1} and \eqref{eq:rate_g2} into $R_{\theta}(T)$ yields \eqref{eq:generate_rate}, which finishes the proof. 
\end{proof}

\subsection{Proof of Theorem \ref{th:state_bound}}\label{pr:th:state_bound}
\begin{proof}
First, based on \eqref{eq:x_xi}, $\mathbb{E}[\|x_{\xi,t} - x^*_t \|^2] $ is written as  
\begin{align}\label{eq:expand_cor_error}
\mathbb{E}[\|x_{\xi,t} \!-\! x^*_t \|^2] \!= & \operatorname{tr}( \mathbb{E}[ (x_{\xi,t} \!-\! x^*_t)(x_{\xi,t} \!-\! x^*_t)^\mathsf{T})])  \nonumber \\
= & n \sigma_{t-1}^2 \!+\! \sum\limits_{m = 0}^{t-2} \operatorname{tr}(\tilde{W}_{t,m} \tilde{W}_{t,m}^\mathsf{T}) \sigma_m^2  \nonumber \\
= & n \sigma_{t-1}^2 \!+\!\! \sum\limits_{m = 0}^{t-2} \!\! \|\tilde{W}_{t,m}\|_F^2 \sigma_m^2. 
\end{align}
where the properties $\operatorname{tr}(\tilde{W}_{t,m} \tilde{W}_{t,m}^\mathsf{T})=\|\tilde{W}_{t,m} \|_F^2$ is applied. 
Since $\mathop {\lim }\limits_{t \to \infty}   n \sigma_{t-1}^2=0$ always holds when $\alpha>0$, we only need to prove $\mathop {\lim }\limits_{t \to \infty} \sum\nolimits_{m = 0}^{t-2} \!\! \|\tilde{W}_{t,m}\|_F^2 \sigma_m^2$ in the sequel. 

Similar to the analysis of $\|W^m \|_F $ in the proof of Theorem \ref{th:equivalent_rate}, 
under Assumption \ref{assu:topo}, we directly consider the eigenvalues of $W$ are distinct for simple analysis.  
Then, recalling the eigenvalue decomposition $W =M \operatorname{diag}\{\lambda_1,\lambda_2,\cdots\!,\lambda_n\} M^{-1} $ and $\lambda_1\!=\!1$, $\tilde{W}_{t,m}$ is equivalent to  
\begin{align}\label{eq:w_m1}
\tilde{W}_{t,m}\!=&M \tilde{J}_{t,m} M^{-1} \nonumber \\
\!=&M\operatorname{diag}\{0,\lambda_i^{t\!-\!m\!-\!1}\!-\!\lambda_i^{t\!-\!m\!-\!2},i\!=\!2,\cdots\!,n\} M^{-1}. 
\end{align}
When the index $m$ is fixed, the non-zero elements of $\tilde{J}_{t,m}$ and $\tilde{W}_{t,m}$ will converge to zero exponentially as $t\to\infty$. 
Note that this exponential convergence still holds even if $ \{\lambda_i,i=2,\!\cdots\!,n\}$ are not distinct (see \cite[Theorem 2.7]{FB-LNS}). 
Hence, there exists a bounded $c_{\epsilon}\!>\!0$ and an arbitrary small $\epsilon\!>\!0$ such that $\|\tilde{W}_{t,m}\|_F^2$ is upper bounded by 
\begin{align}
\|\tilde{W}_{t,m}\|_F^2\le c_{\epsilon}(1-\epsilon)^{t-m}. 
\end{align}
Then, we have 
\begin{align}\label{eq:upper_com}
&\sum\limits_{m = 0}^{t-2} \|\tilde{W}_{t,m}\|_F^2 \sigma_m^2  < \sum\limits_{m = 0}^{t-2} \sigma_0^2 \frac{(1-\epsilon)^{t} }{(1-\epsilon)^{m}  (m+1)^{\alpha})} \nonumber \\
 = & \sigma_0^2 (1-\epsilon)^{t+1}  \sum\limits_{m = 1}^{t-1} \frac{ \rho_\epsilon^m}{m^{\alpha}}=  \sigma_0^2 (1-\epsilon)^{t+1}  \sum\limits_{m = 1}^{t-1} f_\epsilon(m) ,
\end{align}
where $\rho_\epsilon=\frac{1}{1-\epsilon}>1$ and $f_\epsilon(m)=\frac{ \rho_\epsilon^m}{m^{\alpha}}$. 

Next, we turn to prove that the RHS of \eqref{eq:upper_com} will decay to zero as $t\!\to\!\infty$. 
It is straightforward to verify that the $t$-th order derivative of $f_\epsilon(m)$, $f_\epsilon^{(t)}(m)$, contains $2^t$ terms, given by 
\begin{align}
f_\epsilon^{(t)}(m)= \sum\limits_{\ell= 1}^{2^{t}} f_{\epsilon,\ell}^{(t)}(m). 
\end{align}
Recalling the increment characterization condition \eqref{eq:converge_condition} and letting $m=t$, we have 
\begin{align}
&\frac{f_{\epsilon}^{(2t)}(2t)}{(2 \pi)^{2t}}=\frac{  \sum\limits_{\ell= 1}^{2^{2t}} f_{\epsilon,\ell}^{(2t)}}{(2 \pi)^{2t}} \le   \frac{  \max_{\ell} \{f_{\epsilon,\ell}^{(2t)}\} } {(\pi)^{2t}} \nonumber \\
 \Rightarrow ~& \mathop {\lim }\limits_{t \to \infty}  \frac{f_{\epsilon}^{(2t)}(2t)}{(2 \pi)^{2t}} \le \mathop {\lim }\limits_{t \to \infty} \frac{  \max_{\ell} \{f_{\epsilon,\ell}^{(2t)}\} } {(\pi)^{2t}}=0,
\end{align}
where we have used the fact $\rho_\epsilon$ is arbitrarily close to one and $ \mathop {\lim }\limits_{t \to \infty} \frac{1}{(\alpha+t)!}=0$. 
Hence, based on Corollary \ref{coro:sufficient_noise}, the increment scale of $\sum\nolimits_{m = 1}^{t-1} f_\epsilon(m) $ can be characterized by 
\begin{align}\label{eq:scale_bound_integral}
\sum_{m = 1}^{t-1} f_\epsilon(m) = \bm{O} ( \int_1^t \frac{ \rho_\epsilon^y}{y^{\alpha}} dy ). 
\end{align}

Notice that although there exists no an explicit form for $\int_1^t \frac{ \rho_\epsilon^y}{y^{\alpha}} dy$, we can continuously use $\ell_0\in\mathbb{N}^{+}$ times the integration by parts to obtain
\begin{align}\label{eq:integral_expand}
\int_1^t \! \! \frac{ \rho_\epsilon^y}{y^{\alpha}} d y \!=\! & \underbrace{ \frac{ \rho_\epsilon^t }{ t^{\alpha} {\log\rho_\epsilon}} \! + \!\!\sum\limits_{\ell = 1}^{ \ell_0 } \frac{ \rho_\epsilon^t }{  (\alpha \!+\!\ell\!-\!1)! (\log \rho_\epsilon )^{\ell+1} t^{\alpha+\ell}} \!+\! S_{\epsilon}(\ell_0) }_{F_{\epsilon}(\ell_0)} \nonumber \\
& + \frac{1}{(\alpha+\ell_0)!(\log \rho_\epsilon )^{\ell_0+1}} \int_1^t \frac{ \rho_\epsilon^y}{y^{\alpha+\ell_0+1}} dy ,
\end{align}
where $S_{\epsilon}(\ell_0) \!=\! - \frac{ \rho_\epsilon }{\log\rho_\epsilon}\!-\! \sum\limits_{\ell = 1}^{ \ell_0 } \frac{ \rho_\epsilon }{  (\alpha \!+\!\ell\!-\!1)! (\log \rho_\epsilon )^{\ell \!+\!1}}$. 
Given an arbitrary fixed $\rho_\epsilon$, it follows from the Stirling's approximation of $(\alpha+\ell)!$ that 
$\mathop {\lim }\limits_{ \ell \to \infty}  \frac{1}{(\alpha+\ell)!(\log \rho_\epsilon )^{\ell+1}} = 0$, 
which implies that there always exists a strictly bounded $\ell_0$, such that \
\begin{align}\label{eq:bounded_ell}
C_{\epsilon}(\ell_0)=\frac{1}{(\alpha+\ell_0)!(\log \rho_\epsilon )^{\ell_0+1}}<1. 
\end{align}
Substituting \eqref{eq:bounded_ell} into \eqref{eq:integral_expand} and \eqref{eq:scale_bound_integral} successively, we have
\begin{align}\label{eq:scale_integral}
\int_1^t \frac{ \rho_\epsilon^y}{y^{\alpha}} d y= & F_{\epsilon}(\ell_0) + C_{\epsilon}(\ell_0) \int_1^t \frac{ \rho_\epsilon^y}{y^{\alpha+\ell_0+1}} dy \nonumber \\
<& F_{\epsilon}(\ell_0) + C_{\epsilon}(\ell_0) \int_1^t \frac{ \rho_\epsilon^y}{y^{\alpha}} d y \nonumber \\
\Rightarrow \int_1^t \frac{ \rho_\epsilon^y}{y^{\alpha}} d y \!<\! & \frac{F_{\epsilon}(\ell_0) }{1 \!-\!C_{\epsilon}(\ell_0)} \Rightarrow  \sum_{m = 1}^{t-1} f_\epsilon(m) \!<\! \bm{O} ( F_{\epsilon}(\ell_0) ). 
\end{align}
Also, note that the product $\frac{\sigma_0^2 }{\rho_\epsilon^{t+1}} F_{\epsilon}(\ell_0) $ satisfies
\begin{align}\label{eq:temp_zero}
\mathop {\lim }\limits_{t \to \infty} \frac{\sigma_0^2 }{\rho_\epsilon^{t+1}} F_{\epsilon}(\ell_0)  =&  \mathop {\lim }\limits_{t \to \infty} \sum\limits_{\ell = 1}^{ \ell_0 } \frac{ \sigma_0^2  }{ \rho_\epsilon (\alpha \!+\!\ell\!-\!1)! (\log \rho_\epsilon )^{\ell+1} t^{\alpha+\ell}}    \nonumber \\
 & + \mathop {\lim }\limits_{t \to \infty} (\frac{ \sigma_0^2  }{ \rho_\epsilon t^{\alpha} {\log\rho_\epsilon}} \! + \! \frac{S_{\epsilon}(\ell_0) }{\rho_\epsilon^{t+1}}) = 0.  
\end{align}

Finally, substituting \eqref{eq:scale_integral} and \eqref{eq:temp_zero} into \eqref{eq:upper_com}, we have 
\begin{align}
\mathop {\lim }\limits_{t \to \infty}  \sum\limits_{m = 0}^{t-2} \|\tilde{W}_{t,m}\|_F^2 \sigma_m^2 <  \bm{O} (  \mathop {\lim }\limits_{t \to \infty}  \frac{\sigma_0^2 }{\rho_\epsilon^{t+1}} F_{\epsilon}(\ell_0)  ) = 0 ,
\end{align}
which proves \eqref{eq:zero_convergence}. 
Meanwhile, it is clear to see from \eqref{eq:temp_zero} that the upper bound $F_{\epsilon}(\ell_0)$ is monotonically decreasing with $\alpha$ growing (so does $\sigma_t^2$), and thus a larger $\alpha$ will make $\mathbb{E}[\|x_{\xi,t} \!-\! x^*_t \|^2]$ decay to zero faster. 
The proof is completed. 
\end{proof}

\subsection{Proof of Theorem \ref{th:rate_correlation}}\label{pr:th:rate_correlation}
\begin{proof}
The key of this proof is to characterize the increment scale of $\operatorname{tr}( \mathbb{D}[ \Xi X_{\xi}^\mathsf{T}])$, $\operatorname{tr}(\mathbb{E}[\Xi X_{\xi}^\mathsf{T}])$, and $\operatorname{tr}(\mathbb{E}[X_{\xi} X_{\xi}^\mathsf{T}])$, respectively. 
To begin with, we leverage the inequality $\operatorname{tr}( \mathbb{D}[\Xi X_{\xi}^\mathsf{T}]) =  \mathbb{D}[ \operatorname{tr}( \Xi X_{\xi}^\mathsf{T})]$ by Lemma \ref{le:equivalent_trans}, and determine the term $\operatorname{tr}( \Xi X_{\xi}^\mathsf{T})$ first. 
For simple expressions, we collect all noise vectors $\{\xi_t\}_{t=0}^{T-1}$ as a single vector 
\begin{equation}
\bm{\xi}_{0:T-1}=[\xi_0^\mathsf{T},\xi_1^\mathsf{T},\!\cdots\!,\xi_{T-1}^\mathsf{T}]^\mathsf{T}, \nonumber 
\end{equation}
and the notations $\tilde{W}$ defined by \eqref{eq:tilde_w} and $Q_w=\tilde{W}^\mathsf{T} \tilde{W}$ defined by \eqref{eq:Q_W} will be still used in this proof. 

\begin{itemize}
\item \textit{Part 1: Obtain the explicit form of $\operatorname{tr}( \Xi X_{\xi}^\mathsf{T} )$. }
\end{itemize}

First, note that $\bm{\xi}_{0:T-1}$ can be represented by $\bm{\theta}_T$ in the following linear transformation  form of $\bm{\theta}_{0:T-1}$
\begin{align}\label{eq:expre2}
\bm{\xi}_{0:T-1}= H \bm{\theta}_{0:T-1},
\end{align}
where $H \in\mathbb{R}^{nT \times nT}$ is a matrix with $T^2$ blocks of size $n\times n$, and the $(t_1,t_2)$-th block ($t_1,t_2\in\{1,2,\cdots,T\}$) is given by 
\begin{align}
H (t_1,t_2)=\left\{\begin{aligned}
&I,&&\text{if}~t_1=t_2 \\
&-I,&&\text{if}~t_1=t_2+1 \\
&0,&&\text{otherwise}
\end{aligned} \right. .
\end{align}
Based on \eqref{eq:expre2}, $\operatorname{tr}(\Xi X_{\xi}^\mathsf{T}) $ is rewritten as 
\begin{align}\label{eq:new_variance2}
\operatorname{tr}(\Xi X_{\xi}^\mathsf{T}) &= \operatorname{tr}( \sum_{t = 0}^{T-1}\xi_t x_{\xi,t}^\mathsf{T} )= \sum_{t = 0}^{T-1} \operatorname{tr}( x_{\xi,t}^\mathsf{T} \xi_t ) = \sum_{t = 0}^{T-1} \xi_t^\mathsf{T} x_{\xi,t} \nonumber \\
& = \bm{\xi}_{0:T-1}^\mathsf{T}  \bm{x}_{\xi,0:T-1} = \bm{\theta}_{0:T-1}^\mathsf{T} Q \bm{\theta}_{0:T-1},
\end{align}
where $Q=H^\mathsf{T} \tilde{W}  H$ and is represented in the block form
\begin{small}
\begin{align}\label{eq:sum_Q}
 Q\!\!&= \!\!\begin{bmatrix}
g(W,0) & 0  & 0 & \ldots & 0 \\
g(W,1) & g(W,0) & 0 & \ldots & 0 \\
\vdots & \vdots & \ddots & \vdots & \vdots\\
g(W,T\!-\!2) & g(W,T\!-\!3) & \ddots & g(W,0) & 0 \\
W^{T\!-\!2} \!-\! W^{T\!-\!3} & W^{T\!-\!3} \!-\! W^{T\!-\!4} & \ldots & I  & 0
\end{bmatrix}\!\!,
\end{align}
\end{small}
and the matrix function $g(W,T)$ is given by 
\begin{equation}
g(W,t)=\left\{\begin{aligned}
&-I, &&\text{if}~t=0 \\
&2I-W, &&\text{if}~t=1 \\
&2W^{t\!-\!1} \!-\!W^t \!-\! W^{t\!-\!2}, &&\text{if}~t\ge2 \\
\end{aligned}\right. .
\end{equation}

\begin{itemize}
\item \textit{Part 2: Characterize $\operatorname{tr}(\mathbb{E}[\Xi X_{\xi}^\mathsf{T}]) $ and $\operatorname{tr}( \mathbb{D}[ \Xi X_{\xi}^\mathsf{T}] )$. }
\end{itemize}

First, apply Lemma \ref{le:exp_var} on $\mathbb{E}[\operatorname{tr}(\Xi X_{\xi}^\mathsf{T})]$ and we obtain 
\begin{align}\label{eq:scale_expe}
\mathbb{E}[\operatorname{tr}(\Xi X_{\xi}^\mathsf{T})] = &\mathbb{E}[\bm{\theta}_{0:T-1}^\mathsf{T} Q \bm{\theta}_{0:T-1}]= \sum\nolimits_{t= 1}^{T} \operatorname{tr}(Q(t,t)) \sigma_{t-1}^2 \nonumber \\
 =&\sum\nolimits_{t= 1}^{T-1} \operatorname{tr}(g(W,0)) \sigma_{t-1}^2 =-n\sum\nolimits_{t= 0}^{T-2} \sigma_{t}^2,  \nonumber \\
 \Rightarrow & \bm{O}(\mathbb{E}[\operatorname{tr}(\Xi X_{\xi}^\mathsf{T})]) = \bm{O}(\sum\nolimits_{t= 0}^{T-2} \sigma_{t}^2 ). 
\end{align}
Then, applying Lemma \ref{le:exp_var} on $\mathbb{D}[\operatorname{tr}(\Xi X_{\xi}^\mathsf{T})]$, it follows that
\begin{align}\label{eq:overall_sum}
\!\!\mathbb{D}[\operatorname{tr}(\Xi X_{\xi}^\mathsf{T})]= & \underbrace{ \sum_{t= 1}^{T} \sum_{i= 1}^{n} Q_{ii}^2(t,t) \left(\mathbb{E}[(\theta_{t-1}^i)^4]-\mathbb{D}^2[\theta_{t-1}^i] \right) }_{S_1} \nonumber \\
& \!+ \!\!\! \underbrace{ \sum_{t_1,t_2=1}^{T} \sum_{i,j=1}^{n}\!\! Q_{ij}^2(t_1,t_2) \mathbb{D}[\theta_{t_1 \!-\!1}^i] \mathbb{D}[\theta_{t_2\!-\!1}^j] }_{S_2},
\end{align}
where we have used the fact that the non-diagonal elements of $Q(t,t)$ are zeros.  
It is easy to see that the magnitude of matrix blocks $\{Q(t_1,t_2)\}_{t_1,t_2=1}^{T}$ determine the increment scale of $\mathbb{D}[\operatorname{tr}(\Xi X_{\xi}^\mathsf{T})]$. 
For the part $S_1$ in \eqref{eq:overall_sum}, note that the fourth moment of $\theta_{t-1}^i$ is upper bounded by $\mathbb{E}[(\theta_{t-1}^i)^4] \le 16\sigma_{t-1}^4 $ (see \cite[Proposition 2.5.2]{Roman2018High}), and thus we have 
\begin{align}\label{eq:sum_s1}
\!\!&0<\mathbb{D}[(\theta_{t-1}^i)^2] =\left(\mathbb{E}[(\theta_{t-1}^i)^4]-\mathbb{D}^2[\theta_{t-1}^i] \right)\le 14\sigma_{t-1}^4 \Rightarrow \nonumber \\
\!\!& \!\!\!S_1 \!=\! \bm{O} \left( \sum_{t= 1}^{T} \sum_{i= 1}^{n} Q_{ii}^2(t,t)\mathbb{D}[(\theta_{t-1}^i)^2] \right) \! \!=\! \bm{O}\left( \sum_{t= 0}^{T-1} \sigma_t^4 \right)\!. \!\!
\end{align}

Concerning the part $S_2$ in \eqref{eq:overall_sum}, note that the blocks $\{Q(t_1,t_2), t_2\ge t_1+1\}$ are zero matrix blocks according to \eqref{eq:sum_Q}, and thus $S_2$ can be further written as 
\begin{align}\label{eq:sum_s2}
S_{2} \!=& \! \sum_{t_1=2}^{T}\sum_{t_2=1}^{t_1-1} \| Q(t_1,t_2)  \|_F^2 \sigma_{t_1 \!-\!1}^2 \sigma_{t_2 \!-\!1}^2, \nonumber \\
 \!=\!& \sum_{t_1=2}^{T} \sigma_{t_1 \!-\!1}^2 \underbrace{  \left(  \sum_{t_2=1}^{t_1-1} \| Q(t_1,t_2)  \|_F^2  \sigma_{t_2 \!-\!1}^2 \right)}_{S_{2,t_1}}. 
\end{align}
By the construction of \eqref{eq:sum_Q}, most of matrix blocks  $\{Q(t_1,t_2), t_2< t_1\}$ can be represented by $g(W,t)$. 
By the Jordan decomposition $W=MJM^{-1}$, $g(W,t)$ for $t\ge2$ is equivalent to  
\begin{align}
\!\!g(W,t)\!=\!M\operatorname{diag}\{2 \lambda_i^{t\!-\!1} \!-\! \lambda_i^{t} \!-\! \lambda_i^{t\!-\!2}, i=1,\!\cdots\!,n\}M^{-1}. 
\end{align}
Since the sum of the element squares in $g(W,t)$ can be written as $\sum_{i,j=1}^{n} g_{ij}^2(W,t)=\|g(W,t)\|_F^2$, and $|\lambda_i|\le1$ for all $i\in\mathcal{V}$, $\|g(W,t)\|_F^2$ is bounded by 
\begin{align}
\!\!\! \left\{\begin{aligned}
&\! \| g(W,t) \|_F^2 \!\le\!  n^2  \kappa^2(M) \max_i\{ |2 \lambda_i^{t\!-\!1} \!-\! \lambda_i^{t} \!-\! \lambda_i^{t\!-\!2}|^2 \} \\
& \! \| g(W,t) \|_F^2 \!\ge\! \rho^2(g(W,t))\!=\!\max_i\{ |2 \lambda_i^{t\!-\!1} \!-\! \lambda_i^{t} \!-\! \lambda_i^{t\!-\!2}|^2 \}
\end{aligned}\right.\!.
\end{align}
It is clear that $\| g(W,t) \|_F^2 $ will converge to zero exponentially when $t\!\to\!\infty$. 
Similarly, the sum of element squares $\| W^{t\!-\!1} \!-\! W^{t\!-\!2} \|_F^2$ will also converge to zero exponentially as $t\to\infty$. 
Based on this exponential convergence, there exists a $\rho_s\!\in\!(0,1)$ such that 
\begin{align}
\| Q(t_1,t_2)  \|_F^2 = \| g(W,t_1-t_2) \|_F^2=\bm{O}(\rho_s^{t_1-t_2}), 
\end{align}
whose convergence resembles that of  $ \|\tilde{W}_{t,m}\|_F^2$ in \eqref{eq:expand_cor_error}. 
Hence, similar to the analysis of $\sum\limits_{m = 0}^{t-2} \!\! \|\tilde{W}_{t,m}\|_F^2 \sigma_m^2$, $S_{2,t_1}$ satisfies $\mathop {\lim }\limits_{t_1 \to \infty} S_{2,t_1}=0 $, 
and $ \bm{O}(S_{2,t_1})=\sum\limits_{t_2 = 0}^{t_1-1} \frac{\rho_s^{t_1-t_2}}{(t_2+1)^\alpha} $ can be upper bounded by 
\begin{align}
\sum\limits_{t_2 = 0}^{t_1-1} \frac{\rho_s^{t_1-t_2}}{(t_2+1)^\alpha} \le \rho_s^{t_1-1} + \sum\limits_{t_2 = 1}^{t_1-1} \frac{\rho_s^{t_1-t_2}}{2^\alpha}.
\end{align}
Then, it further follows that 
\begin{align}\label{eq:S2_scale2}
0 \le \bm{O}(S_{2,t_1}) \le \bm{O}{(\frac{1}{2^\alpha})} \Rightarrow \bm{O}(S_2) \le \bm{O}({\frac{1}{2^\alpha} \sum_{t=1}^{T-1} \sigma_{t}^2 } ). 
\end{align}
Finally, summing up \eqref{eq:sum_s1} and \eqref{eq:S2_scale2}, we have 
\begin{align}\label{eq:final_dxx}
\!\!\!\! \bm{O}( \operatorname{tr}(\mathbb{D}[\Xi X_{\xi}^\mathsf{T}]) )\!=\! \bm{O}( \mathbb{D}[\operatorname{tr}(\Xi X_{\xi}^\mathsf{T})] )\!\le\! \bm{O} \left( \frac{1}{2^\alpha} \sum\limits_{t = 0}^{T\!-\!1} \sigma_t^2 \right) ,
\end{align}
where we have kept the term $\sigma_0^4$ in $S_1$ and omitted the rest $\sum\nolimits_{t = 1}^{T-1} \sigma_t^4$ therein, because the increment scale of $\sum\nolimits_{t = 1}^{T-1} \sigma_t^4$ is always smaller than that of $\sum\nolimits_{t = 1}^{T-1} \sigma_t^2$ when $\alpha\ge0$.


\begin{itemize}
\item \textit{Part 3: Characterize $\bm{O}( \operatorname{tr}(\mathbb{E}[X_{\xi} X_{\xi}^\mathsf{T}]) )$. }
\end{itemize}

First, the term $\operatorname{tr}(X_{\xi} X_{\xi}^\mathsf{T})$ can be represented by
\begin{align}\label{eq:variance_cc}
\operatorname{tr}(X_{\xi} X_{\xi}^\mathsf{T}) = & \operatorname{tr}( \sum_{t = 0}^{T-1} x_{\xi} x_{\xi}^\mathsf{T} ) = \bm{x}_{\xi,0:T-1}^\mathsf{T}  \bm{x}_{\xi,0:T-1}  \nonumber \\
 = & \bm{\theta}_{0:T\!-\!1}^\mathsf{T} H^\mathsf{T} Q_w  H \bm{\theta}_{0:T\!-\!1} \!=\! \bm{\theta}_{0:T\!-\!1}^\mathsf{T} \tilde{Q} \bm{\theta}_{0:T\!-\!1},
\end{align}
where $\tilde{Q}=H^\mathsf{T} \tilde{W}^\mathsf{T} \tilde{W}  H$. 
Specifically, the $(t,t)$-th diagonal block of $\tilde{Q}$ is written as 
\begin{align}\label{eq:vm}
\tilde{Q}(t,t) = & Q_w(t,t) \!+\! Q_w(t\!+\!1,t\!+\!1) \!-\!Q_w(t\!+\!1,t) \!-\! Q_w(t,t\!+\!1) \nonumber \\
=&\Gamma_{T-t}^*+ \Gamma_{T-t-1}^*- \Gamma_{1,T-t-1}^*(1) - \Gamma_{1,T-t-1}^{*\mathsf{T}}(1) \nonumber \\
= & I \!+\!\! \sum_{m = 0}^{T-t-1} \!\! (W^{m\!+\!1}\!-\!W^{m})^\mathsf{T} (W^{m\!+\!1}\!-\!W^{m}) 
\end{align}
for $t\!\in\!\{1,\!\cdots\!,T\!-\!1\}$ and $\tilde{Q}(T,T)=I$. 
Similar to the analysis of $g(W,t)$, it is straightforward to obtain that the elements of $\tilde{Q}(t,t)$ are strictly bounded due to $\rho(W)=1$ and the fact that $(W^{m}\!-\!W^{m+1})$ will converge to zero exponentially as $m\to\infty$. 
Therefore, there exists a constant $C_2$ such that 
\begin{align}\label{eq:bound_exx}
n \le \operatorname{tr}(\tilde{Q}(t,t)) \le C_2<\infty,~\forall t\in \mathbb{N}^+. 
\end{align} 
Then, based on Lemma \ref{le:exp_var} and \eqref{eq:bound_exx}, we have that 
\begin{align}\label{eq:rate:exx}
&\operatorname{tr}(\mathbb{E}[X_{\xi} X_{\xi}^\mathsf{T}])  = \mathbb{E}[\bm{\theta}_{0:T\!-\!1}^\mathsf{T} \tilde{Q} \bm{\theta}_{0:T\!-\!1}] =  \sum\nolimits_{t = 1}^{T} \operatorname{tr}(\tilde{Q}(t,t)) \sigma_{t-1}^2 \nonumber \\  
& \Rightarrow ~ \bm{O}( \operatorname{tr}(\mathbb{E}[X_{\xi} X_{\xi}^\mathsf{T}]) )= \bm{O}( \sum\nolimits_{t = 0}^{T-1} \sigma_{t}^2 ).
\end{align}

\begin{itemize}
\item \textit{Part 4: Characterize the decaying rate of $R_{\xi}(T)$.}
\end{itemize}

Finally, by the definition of $R_{\xi}(T)$ and utilizing the increment scale characterization \eqref{eq:scale_expe}, \eqref{eq:final_dxx} and \eqref{eq:rate:exx}, we have 
\begin{align}\label{eq:final_rate_cor}
R_{\xi}(T) & =\bm{O}\left( \frac{ \operatorname{tr}(\mathbb{E}[\Xi X_{\xi}^\mathsf{T}]) \pm c_{\sigma} \sqrt{\operatorname{tr}( \mathbb{D}[ \Xi X_{\xi}^\mathsf{T}])}}{\operatorname{tr}(\mathbb{E}[X_{\xi} X_{\xi}^\mathsf{T}])} \right) \nonumber \\
& = \frac{ \bm{O}(\sum\nolimits_{t= 0}^{T-1} \sigma_{t}^2 ) \pm c_{\sigma}  \bm{O}(\sqrt{ \frac{1}{2^\alpha}\sum\nolimits_{t= 0}^{T-1} \sigma_{t}^2}) }{ \bm{O}(\sum\nolimits_{t= 0}^{T-1} \sigma_{t}^2 ) }  \nonumber \\
& = \bm{O}(C) \pm c_{\sigma}\bm{O} \left(\frac{1}{\sqrt{ {2^\alpha} \int_1^T \frac{1}{y^\alpha} d{y} }} \right),
\end{align}
which further yields the decaying rate \eqref{eq:rates_cor} for different $\alpha$, and the coefficient $\frac{1}{2^\alpha}$ is omitted when $\alpha\le1$ because $\int_1^T \frac{1}{y^\alpha} d{y}$ goes to infinity when $T$ grows. 
The proof is completed. 
\end{proof}

\subsection{Proof of Corollary \ref{coro:multi}}\label{pr:coro:multi}
\begin{proof}
The proof of this result resembles that of Theorem \ref{th:state_bound} and \ref{th:rate_correlation}, and the sketch is provided here. 

First, for the convergence of the state deviation, substituting \eqref{eq:new_xi2} into $x_{\xi,t} - x^*_t$ ($t>k$), we have 
\begin{small}
\begin{align}
&x_{\xi,t} - x^*_t \nonumber \\
= &  \sum\limits_{m = k+1}^{t-1} \!\! W^{t-m-1}  (  \sum\limits_{\ell = 1}^{k+1} p_{\ell} \theta_{m-\ell} ) + \!\! \sum\limits_{m = 0}^{k}   W^{t-m-1} (  \sum\limits_{\ell = 1}^{m+1} p_{\ell} \theta_{m-\ell} ) \nonumber \\
=& \sum\limits_{m = 0}^{t-k}  \tilde{W}_{p,a}(t,m) \theta_m + \sum\limits_{m = t-k+1}^{t-1}\tilde{W}_{p,b}(t,m) \theta_m ~~ \Rightarrow \nonumber \\
& \!\! \! \mathbb{E}[\|x_{\xi,t} - x^*_t\|^2]= \!\! \sum\limits_{m = 0}^{t-k} \!\! \|\tilde{W}_{p,a}(t,m)\|_F^2 \sigma_m^2 + \! \!\!\!\!\!\!\sum\limits_{m = t-k+1}^{t-1} \!\!\!\!\!\!\!\! \|\tilde{W}_{p,b}(t,m)\|_F^2 \sigma_m^2,
\end{align}
\end{small}
\!\!\!where $ \tilde{W}_{p,a}(t,m)= \sum_{\ell = 1}^{k} p_{\ell} W^{t-m-\ell} $ and $\tilde{W}_{p,b}(t,m) = \sum_{\ell = 1}^{t-m} p_{\ell} W^{t-m-\ell}$. 
By the condition \eqref{eq:condition_p} and the eigenvalue decomposition of $\tilde{W}_{p,a}(t,m)$, we have 
\begin{align}\label{eq:extend_conv}
\!\!\!\! \|\tilde{W}_{p,b}(t,m)\|_F^2 \!\le \! (k\!-\!1)n\bar{p}^2,~\mathop {\lim }\limits_{t \to \infty} \|\tilde{W}_{p,a}(t,m)\|_F\!=\!0,
\end{align}
where the latter one converges to zero exponentially. 
Based on \eqref{eq:extend_conv}, the convergence of $\mathbb{E}[\|x_{\xi,t} - x^*_t\|^2]$ is proved by following the same procedures of proving $\mathop {\lim }\limits_{t \to \infty}  \sum\limits_{m = 0}^{t-2} \|\tilde{W}_{t,m}\|_F^2 \sigma_m^2 =0$ in Appendix \ref{pr:th:state_bound}.

Next, for the convergence of $R_\xi(T)$ when $\xi_t$ subjects to \eqref{eq:new_xi2}, we highlight the key point is to establish the expressions of $\bm{x}_{\xi,0:T-1}$ and $\bm{\xi}_{0:T-1}$ about $\bm{\theta}_{0:T-1}$, given by
\begin{align}
\bm{x}_{\xi,0:T-1} =  \tilde{W} H_k \bm{\theta}_{0:T-1},~\bm{\xi}_{0:T-1}=H_k \bm{\theta}_{0:T-1},
\end{align}
where $\tilde{W} $ is defined \eqref{eq:tilde_w}, 
and $H_k\in\mathbb{R}^{nT \times nT}$ is a matrix with $T^2$ blocks of size $n\times n$, and the $(t_1,t_2)$-th block ($t_1,t_2\in\{1,2,\cdots,T\}$) is given by 
\begin{align}
H_k(t_1,t_2)\!=\!\left\{\begin{aligned}
& p_{t_1-t_2+1} I,&&\text{if}~t_2 \le t_1\le k \\
& p_{k-t_2+1} I,&&\text{if}~k<t_1 ~\text{and}~0\!\le\! t_1\!-\!t_2\le k\!-\!1 \\
&0,&&\text{otherwise}
\end{aligned} \right. .
\end{align}
Then, by the same proof procedures of Theorem \ref{th:rate_correlation}, it can be verified that $Q_k=H_k^\mathsf{T} \tilde{W}  H_k$ and $\tilde{Q}_k=H_k^\mathsf{T} \tilde{W}^\mathsf{T} \tilde{W}  H_k$ satisfy
\begin{align}
\sum_{t_1,t_2=1}^{T} \!\!\!\! \| Q_k(t_1,t_2)  \|_F^2 \!=\! \bm{O}\left(k \sum_{t= 0}^{T-1} \sigma_t^2 \right),~\operatorname{tr}(\tilde{Q}_k (t,t) = \bm{O}(C) ,
\end{align}
which are used to derive that the increment scales of $\operatorname{tr}(\mathbb{E}[\Xi X_{\xi}^\mathsf{T}]) $, $\operatorname{tr}( \mathbb{D}[ \Xi X_{\xi}^\mathsf{T}] )$ and $\operatorname{tr}(\mathbb{E}[X_{\xi} X_{\xi}^\mathsf{T}])$ are the same as those when $\xi_t=\theta_t - \theta_{t-1}$ in Theorem \ref{th:rate_correlation}. 
The procedures resemble those in Appendix \ref{pr:th:rate_correlation} and are omitted here. 
\end{proof}

\bibliographystyle{IEEEtran}

\begin{IEEEbiographynophoto}{Yushan Li}
(S'19) received the B.E. degree in automatic control from the School of Artificial Intelligence and Automation, Huazhong University of Science and Technology, Wuhan, China, in 2018. He is currently working toward the Ph.D. degree in control science and engineering with the Department of Automation, Shanghai Jiaotong University, Shanghai, China. 
He is a member of Intelligent of Wireless Networking and Cooperative Control Group. His research interests include robotics, security of cyber-physical system, and distributed computation and optimization in multiagent networks.
\end{IEEEbiographynophoto}

\begin{IEEEbiographynophoto}{Zitong Wang }
(S'21) received the B.E. degree in information engineering from the School of Electronic Information and Electrical Engineering at Shanghai Jiao Tong University, China, in 2020. 
She is currently pursuing the Ph.D. degree in control science and engineering at the Department of Automation, Shanghai Jiao Tong University. Her research interests include robotics, cooperative control, and distributed optimization in multi-agent systems.
\end{IEEEbiographynophoto}

\begin{IEEEbiographynophoto}{Jianping He} 
(SM'19) is currently an associate professor in the Department of Automation at Shanghai Jiao Tong University. He received the Ph.D. degree in control science and engineering from Zhejiang University, Hangzhou, China, in 2013, and had been a research fellow in the Department of Electrical and Computer Engineering at University of Victoria, Canada, from Dec. 2013 to Mar. 2017. His research interests mainly include the distributed learning, control and optimization, security and privacy in network systems.

Dr. He serves as an Associate Editor for IEEE Trans. Control of Network Systems, IEEE Open Journal of Vehicular Technology, and KSII Trans. Internet and Information Systems. He was also a Guest Editor of IEEE TAC, International Journal of Robust and Nonlinear Control, etc. He was the winner of Outstanding Thesis Award, Chinese Association of Automation, 2015. He received the best paper award from IEEE WCSP'17, the best conference paper award from IEEE PESGM'17, and was a finalist for the best student paper award from IEEE ICCA'17, and the finalist best conference paper award from IEEE VTC'20-FALL.
\end{IEEEbiographynophoto}

\begin{IEEEbiographynophoto}{Cailian Chen}
(M'06) received the B.E. and M.E. degrees in Automatic Control from Yanshan University, P. R. China in 2000 and 2002, respectively, and the Ph.D. degree in Control and Systems from City University of Hong Kong, Hong Kong SAR in 2006. She joined Department of Automation, Shanghai Jiao Tong University in 2008 as an Associate Professor. She is now a Full Professor. 
Before that, she was a postdoctoral research associate in University of Manchester, U.K. (2006-2008). 
She was a Visiting Professor in University of Waterloo, Canada (2013-2014). 
Prof. Chen's research interests include industrial wireless networks, computational intelligence and situation awareness, Internet of Vehicles. 

Prof. Chen has authored 3 research monographs and over 100 referred international journal papers. She is the inventor of more than 20 patents. 
She received the prestigious ”IEEE Transactions on Fuzzy Systems Outstanding Paper Award” in 2008, and Best Paper Award of WCSP17 and YAC18. 
She won the Second Prize of National Natural Science Award from the State Council of China in 2018, First Prize of Natural Science Award from The Ministry of Education of China in 2006 and 2016, respectively, and First Prize of Technological Invention of Shanghai Municipal, China in 2017. She was honored Changjiang Young Scholar in 2015 and Excellent Young Researcher by NSF of China in 2016. 
Prof. Chen has been actively involved in various professional services. 
She serves as Associate Editor of IEEE Transactions on Vehicular Technology, Peerto-peer Networking and Applications (Springer). 
She also served as Guest Editor of IEEE Transactions on Vehicular Technology, TPC Chair of ISAS19, Symposium TPC Co-chair of IEEE Globecom 2016 and VTC2016-fall, Workshop Co-chair of WiOpt18.
\end{IEEEbiographynophoto}

\begin{IEEEbiographynophoto}{Xinping Guan}
(F'18) received the B.S. degree in Mathematics from Harbin Normal University, Harbin, China, in 1986, and the Ph.D. degree in Control Science and Engineering from Harbin Institute of Technology, Harbin, China, in 1999. He is currently a Chair Professor with Shanghai Jiao Tong University, Shanghai, China, where he is the Dean of School of Electronic, Information and Electrical Engineering, and the Director of the Key Laboratory of Systems Control and Information Processing, Ministry of Education of China. 
Before that, he was the Professor and Dean of Electrical Engineering, Yanshan University, Qinhuangdao, China. 

Dr. Guan's current research interests include industrial cyber-physical systems, wireless networking and applications in smart factory, and underwater networks. He has authored and/or coauthored 5 research monographs, more than 270 papers in IEEE Transactions and other peer-reviewed journals, and numerous conference papers. 
As a Principal Investigator, he has finished/been working on many national key projects. He is the leader of the prestigious Innovative Research Team of the National Natural Science Foundation of China (NSFC). 
Dr. Guan is an Executive Committee Member of Chinese Automation Association Council and the Chinese Artificial Intelligence Association Council. 
Dr. Guan received the First Prize of Natural Science Award from the Ministry of Education of China in both 2006 and 2016, and the Second Prize of the National Natural Science Award of China in both 2008 and 2018. 
He was a recipient of IEEE Transactions on Fuzzy Systems Outstanding Paper Award in 2008. He is a National Outstanding Youth honored by NSF of China, Changjiang Scholar by the Ministry of Education of China and State-level Scholar of New Century Bai Qianwan Talent Program of China.
\end{IEEEbiographynophoto}

\end{document}